\documentclass[onecolumn]{IEEEtran}
\usepackage{cite}
\usepackage{amsmath}
\usepackage{amsfonts}
\usepackage{pifont}
\usepackage{amssymb}
\usepackage{amsthm}
\usepackage{tikz}
\usepackage{cases}
\usepackage{graphicx}
    \graphicspath{{../}}
    \DeclareGraphicsExtensions{.pdf}
\usepackage[justification=centering]{caption}
\usepackage[caption=false,font=footnotesize]{subfig}
\usepackage{multirow}
\usepackage{booktabs}
\usepackage{url}
\usepackage{xtab}
\usepackage{tabu}
\usepackage{longtable}
\usepackage{algorithm}
\usepackage{algorithmic}
\usepackage{enumerate}
\usepackage{makecell}
\usepackage{lipsum}
\usepackage{multicol}
\usepackage{mathdots}
\usepackage{extarrows}
\usepackage{color,xcolor}
\usepackage{bm} 
\usetikzlibrary{arrows}
\theoremstyle{plain}

\newtheorem{theorem}{Theorem}
\newtheorem{lemma}[theorem]{Lemma}

\newtheorem{definition}[theorem]{Definition}
\newtheorem{corollary}[theorem]{Corollary}
\theoremstyle{definition}

\newtheorem{remark}{Remark}

\begin{document}
\title{Scalar MSCR Codes via the Product Matrix Construction}
\author{\IEEEauthorblockN{Yaqian Zhang, ~Zhifang Zhang}\\
\IEEEauthorblockA{\fontsize{9.8}{12}\selectfont KLMM, Academy of Mathematics and Systems Science, Chinese Academy of Sciences, Beijing 100190, China\\
 School of Mathematical Sciences, University of Chinese Academy of Sciences, Beijing 100049, China\\
Emails: zhangyaqian15@mails.ucas.ac.cn, ~zfz@amss.ac.cn}

}

\maketitle
\thispagestyle{empty}
\begin{abstract}
An $(n,k,d)$ cooperative regenerating code provides the optimal-bandwidth repair for any $t~(t\!>\!1)$ node failures in a cooperative way. In particular, an MSCR (minimum storage cooperative regenerating) code retains the same storage overhead as an $(n,k)$ MDS code. Suppose each node stores $\alpha$ symbols which indicates the sub-packetization level of the code. A scalar MSCR code attains the minimum sub-packetization, i.e., $\alpha=d-k+t$. By now, all existing constructions of scalar MSCR codes restrict to very special parameters, eg. $d=k$ or $k=2$, etc. In a recent work, Ye and Barg construct MSCR codes for all $n,k,d,t$, however, their construction needs $\alpha\approx{\rm exp}(n^t)$ which is almost infeasible in practice. In this paper, we give an explicit construction of scalar  MSCR codes for all $d\geq \max\{2k-1-t,k\}$, which covers all possible parameters except the case of $k\leq d\leq 2k-2-t$ when $k<2k-1-t$. Moreover, as a complementary result, for $k<d<2k-2-t$ we prove the nonexistence of linear scalar MSCR codes that have invariant repair spaces. Our construction and most of the previous scalar MSCR codes all have invariant repair spaces and this property is appealing in practice because of convenient repair. As a result, this work presents an almost full description of linear scalar MSCR codes.

\end{abstract}

\section{Introduction}\label{sec1}
A central issue in large-scale distributed storage systems (DSS) is the efficient repair of node failures. Suppose a data file is stored across $n$ nodes such that a data collector can retrieve the original file by reading the contents of any $k$ nodes. When some node fails, a self-sustaining storage system tends to regenerate the failed node by downloading data from some surviving nodes (i.e. {helper nodes}), which is called the {node repair} problem. An important metric for the node repair efficiency is the {repair bandwidth}, namely, the total amount of data downloaded during the repair process. In a celebrated work \cite{Dimakis2011}, Dimakis et al. proposed the {regenerating codes} which achieve the tradeoff (i.e. the cut-set bound) between the repair bandwidth and the amount of data stored per node. That is, fixing the repair bandwidth,
the storage  cannot be further reduced, and vice versa. In particular, regenerating codes with the minimum storage and with the minimum repair bandwidth are respectively called MSR codes and MBR codes. Constructing MSR codes and MBR codes for general parameters has been extensively studied in a series of works \cite{Ramchandran2011,Kumar2011,Vardy2016,Sasidharan2015,Ye2016,Ye2016sub-}.

Regenerating codes typically deal with single node failures, however, the scenarios of multiple node failures are quite common in DSS.  For example, in Total Recall
\cite{TotalRecall} a repair process is triggered only after the total number of failed nodes has reached a predefined threshold. There are two typical models for repairing multiple node failures.  One is the centralized repair model where a special node called data center is assumed to complete all repairs. That is, suppose $t$ nodes fail, then the data center downloads data from helper nodes and is responsible for generating all the $t$ new nodes. The other model is the cooperative repair where $t$ new nodes are generated in a distributed and cooperative way. It was proved in \cite{Ye2018} that an MDS code achieving the optimal bandwidth with cooperative repair must also have optimal bandwidth with centralized repair. Therefore, cooperative repair can be viewed as a harder problem than the centralized repair. Moreover, due to the distributed pattern, the cooperative repair fits DSS better than the centralized repair. Regenerating codes with centralized repair are studied in \cite{Cadambe2013,Ye2016,Tamo2016}. In this paper, we focus on regenerating codes with cooperative repair.

The idea of cooperative repair was proposed by Hu et al. in \cite{Hu2010}. Specifically, suppose $t$ new nodes (i.e. newcomers) are to be generated as replacements of $t$ failed nodes respectively. Then the regenerating process is carried out in two phases. Firstly, each newcomer connects to $d$ surviving nodes (i.e. helper nodes) and downloads $\beta_1$ symbols from each. Note that different newcomers
may choose different $d$ helpler nodes; Secondly, each newcomer downloads $\beta_{2}$ symbols from each of the other $t-1$ newcomers.  Therefore the repair bandwidth for repairing one failed node is $\gamma=d\beta_{1}+(t-1)\beta_{2}$. In \cite{Hu2010,Shum2011} a cut-set bound was derived for regenerating codes with cooperative repair, i.e.,
\begin{equation*}
B\leq \sum_{i=1}^{s}l_{i}\min\{\alpha,(d-\sum_{h=1}^{i-1}l_{h})\beta_{1}+(t-l_{i})\beta_{2}\},
\end{equation*}
where $B$ is the size of the original data file, $\alpha$ is the size of data stored in each node,  and $l_1,...,l_s,s$ are the integers satisfying $l_{1}+\cdots+l_{s}=k$ and $1\leq l_{1},...,l_{s}\leq t$. The cut-set bound, along with other necessary conditions on the parameters, define a $\alpha$-$\gamma$ tradeoff curve and codes with parameters lying on this curve are called cooperative regenerating codes. In particular, the two extreme points on the tradeoff curve respectively correspond to MBCR (\emph{minimum bandwidth cooperative regenerating}) and MSCR (\emph{minimum storage cooperative regenerating}) codes.
Parameters of the two types of codes are given below:

For MBCR codes,
\begin{equation*}
\alpha=\gamma=(2d+t-1)\beta_2, \ \ \beta_{1}=2\beta_{2}, \ \ \beta_{2}=\frac{B}{k(2d-k+t)}.
\end{equation*}

For MSCR codes,
\begin{equation}\label{mscr}
\alpha=\frac{B}{k}=(d-k+t)\beta_2, \ \ \beta_{1}=\beta_{2}=\frac{B}{k(d-k+t)}.
\end{equation}
It can be seen that $B, \alpha,\beta_1$ are all integral multiples of $\beta_2$. When $\beta_2=1$, the corresponding code is called a {\it scalar} code. Moreover, $\alpha$ is called the sub-packetization which indicates the sub-block size in the encoding phase. Since the sub-packetization also decides the smallest file size and the number of operations needed in encoding and repairing, regenerating codes with small sub-packetization are preferred. Obviously, a scalar regenerating code has the smallest sub-packetization, however, for some parameters scalar codes cannot exist.

About the constructions of MBCR and MSCR codes, Shum et al.\cite{Shum2011} first explicitly constructed scalar MSCR codes for $d=k$. Then Shum and Hu \cite{Shum&Hu:MBCR} also built MBCR codes for the case of $d=k$ and $n=d+t$. Later, Wang and Zhang \cite{wang2013} gave an explicit construction of scalar MBCR codes for all $n,k,d,t$. However, constructing MSCR codes seems to be a much harder problem. In \cite{Scouarnec2012} the author built scalar MSCR codes for the case $k=2$ and $d=n-t$. Then Chen and Shum\cite{Chen2013} designed a scalar $(n=2k,k,d=2k-2,t=2)$ MSCR code. After that, they generalized the construction and obtained $(n=2k,k,d=n-t,2\leq t\leq n-k)$ scalar MSCR codes but the failed nodes must be systematic nodes \cite{Shum2016}. An obvious drawback of all existing constructions of MSCR codes is that they all restrict to very limited parameters. Until recently, Ye and Barg\cite{Ye2018} gave an explicit construction of MSCR codes for all parameters, i.e., $2\leq t\leq n-k, ~d\geq k+1$. However, the sub-packetization of their codes is extraordinarily large, i.e. {\small $\alpha\!=\!\big((d-k+t)(d-k)^{t-1}\big)^{\binom{n}{t}}\!\approx\!{\rm exp}(n^t)$}. Note that scalar MSCR codes have sub-packetization $\alpha=d-k+t$.

In this paper, we present an explicit construction of scalar MSCR codes for all $d\geq \max\{2k-1-t,k\}$, which almost covers all parameters except for an augment of $d$ when $k>t+1$. Besides, the authors of \cite{Shah2009} have proved nonexistence of linear scalar MSR codes for $k<d<2k-3$, we continue to prove that there exist no linear scalar MSCR codes with invariant repair space for $k<d<2k-2-t$. It's worth noting that our codes presented here are exactly the linear scalar MSCR codes with invariant repair spaces. Thus our construction gives an almost full description of such codes. Moreover, our construction can be viewed as an extension of the product matrix construction for MSR codes proposed by Rashmi et al. \cite{Kumar2011} to cooperative regenerating codes.

The remaining of this paper is organized as follows. First, a brief introduction to the product matrix framework is given in Section II. Then the construction of MSCR codes is presented in three steps respectively in Section \ref{sec3}-\ref{sec5}. Specifically, Section \ref{sec3} deals with the case $t=2$ and $d=2k-3$, Section \ref{sec4} extends to any $2\leq t\leq k-1$ and $d=2k-1-t$, and Section \ref{sec5} further extends to all $2\leq t\leq n-k$ and $d\geq\max\{2k-1-t,k\}$. Section \ref{sec5-c} proves the nonexistence of linear scalar MSCR codes for $k<d<2k-2-t$. Finally, Section \ref{sec6} concludes the paper.

\section{The Product matrix framework}\label{sec2}
Let $[n]$ stand for $\{1,2,...,n\}$. We recall the product matrix framework proposed in \cite{Kumar2011}. First, each codeword is represented by an $n\times\alpha$ matrix $C$ with the $i$-th row stored in the $i$-th node for all $i\in[n]$. Moreover, the codeword is generated as a product of two matrices, i.e.,
$$
C=GM,
$$
where $G$ is an $n\times d$ \emph{encoding matrix}, and $M$ is a $d\times\alpha$ \emph{message matrix}. More specifically, the matrix $M$ contains all message symbols while some entries may be linear combinations of the message symbols, and $G$ is a predefined matrix which is independent of the message. To store a data file (i.e., a message) consisting of $B$ symbols, we firstly arrange the $B$ symbols into the message matrix $M$ properly, and then calculate $C=GM$ to obtain a codeword which will be stored across $n$ nodes. Denote by $\bm{\psi}_{i}$ the $i$-th row of $G$, then the content stored in node $i$ is represented by
$$
\bm{c}_{i}=\bm{\psi}_{i}M,
$$
for $i\in[n]$. Throughout this paper, we use bold letters (eg. $\bm{\psi}, \bm{\varphi}, \bm{c}$, etc.) to denote vectors and capital letters (eg. $C,M,\Phi$, etc.) denote matrices.

Next we use the construction of scalar MSR codes to describe how this framework works for building regenerating codes. For simplicity, we just consider the case $d=2k-2$.

First, the parameters of the scalar MSR code with $d=2k-2$ are as follows:
$$
\alpha=d-k+1=k-1, \ \ \ \beta=1, \ \ \ B=k\alpha=k(k-1).
$$
Let $\alpha_{1},...,\alpha_{n}$ be $n$ distinct nonzero elements in a finite field $\mathbb{F}_q$ such that $\alpha_{i}^{k-1}\neq\alpha_{j}^{k-1}$ for any $ i\neq j$.
Then the $n\times d$ encoding matrix is set to be
\begin{equation}\label{G}
G=\begin{pmatrix}
1 & \alpha_{1} &\cdots & \alpha_{1}^{2k-3} \\
1 & \alpha_{2} &\cdots & \alpha_{2}^{2k-3} \\
\vdots & \vdots &\vdots &\vdots \\
1 & \alpha_{n} &\cdots & \alpha_{n}^{2k-3}
\end{pmatrix}.
\end{equation}
The $d\times\alpha$ message matrix is defined as
$$
M=\begin{pmatrix}
S_{1} \\
S_{2}
\end{pmatrix},
$$
where $S_{1}$ and $S_{2}$ are two $(k-1)\times(k-1)$ symmetric matrices each filled with $\frac{k(k-1)}{2}$ message symbols. Thus the total number of symbols contained in $M$ is $k(k-1)=B$. For $i\in[n]$, the $i$-th node stores a row vector of length $\alpha=k-1$:
$$
\bm{c}_{i}=\bm{\psi}_{i}M=\begin{pmatrix} 1&\alpha_{i}&\cdots&\alpha_{i}^{2k-3} \end{pmatrix}M.
$$

Then we illustrate the above construction gives a scalar MSR code. Actually, we need to verify the following two properties:
\begin{enumerate}
  \item {\it Node repair: }Suppose the $i$-th node fails, then a newcomer can recover the content stored in the $i$-th node by connecting to any $d$ helper nodes and downloading $\beta=1$ symbol from each.

  Let $R\subseteq[n]\setminus\{i\}$ be the set of $d$ helper nodes and define a row vector
$\bm{\varphi}_{i}=(1, \alpha_{i}, \cdots, \alpha_{i}^{k-2}
)$. Then each node $j\in R$ sends the symbol
$$
\bm{c}_{j}\bm{\varphi}_{i}^\tau=\bm{\psi}_{j}M\bm{\varphi}_{i}^\tau
$$
to the newcomer, where $\tau$ denotes the transpose. Thus the newcomer obtains the symbols: $\Psi_{repair}M\bm{\varphi}_{i}^\tau$, where $\Psi_{repair}$ is the matrix $G$ restricted to the rows indexed by the elements in $R$. Since $\Psi_{repair}$ is a $d\times d$ Vandermonde matrix, by multiplying the inverse of $\Psi_{repair}$ the newcomer obtains
$$
M\bm{\varphi}_{i}^\tau=\begin{pmatrix}
S_{1}\bm{\varphi}_{i}^\tau \\ S_{2}\bm{\varphi}_{i}^\tau
\end{pmatrix}.
$$
Since $S_{1}$ and $S_{2}$ are both symmetric, then $\bm{\varphi}_{i}{S_{1}}$ and $\bm{\varphi}_{i}{S_{2}}$ are obtained by a transposition. Finally, the newcomer computes  $\bm{\varphi}_{i}{S_{1}}+\alpha_{i}^{k-1}\bm{\varphi}_{i}{S_{2}}=\bm{\psi}_{i}M=\bm{c}_i$ which are exactly the contents stored in the $i$-th node.

\vspace{5pt}
  \item {\it Data reconstruction: }A data collector can recover the original data file by reading the contents stored in any $k$ nodes.

  Suppose the data collector connects to $k$ nodes $i_{1},i_{2},...,i_{k}\in[n]$. Let $\Psi$ be the $k\times d$ sub-matrix of $G$ consisting of the rows $\bm{\psi}_{i_1},...,\bm{\psi}_{i_k}$, then the data collector reads the $k\alpha$ symbols
\begin{equation*}
\begin{aligned}
\Psi M&=\Psi\begin{pmatrix} S_{1} \\ S_{2} \end{pmatrix} \\
&=\Phi S_{1}+\Delta \Phi S_{2},
\end{aligned}
\end{equation*}
where $\Phi$ is the $\!k\!\times\!(k\!-\!1)\!$ matrix consisting of the first $k\!-\!1$ columns of $\Psi$, i.e.,
\begin{equation}\label{phi}
\Phi=\begin{pmatrix}
1 & \alpha_{i_{1}} &\cdots & \alpha_{i_{1}}^{k-2} \\
1 & \alpha_{i_{2}} &\cdots & \alpha_{i_{2}}^{k-2} \\
\vdots & \vdots &\vdots &\vdots \\
1 & \alpha_{i_{k}} &\cdots & \alpha_{i_{k}}^{k-2}
\end{pmatrix},
\end{equation}
and $\Delta$ is a $k\!\times\! k$ diagonal matrix ${\rm diag}[\alpha_{i_{1}}\!^{\!k\!-\!1}, \alpha_{i_{2}}\!^{\!k\!-\!1},...,\alpha_{i_{k}}\!^{\!k\!-\!1} ]$. Obviously, $\Phi$ is a Vandermonde matrix and $\Psi\!=\!(
\Phi~~\Delta\Phi
)$.  The data collector can recover the data file due to the following lemma which is also a result in \cite{Kumar2011}. We leave the proof of this lemma in Appendix \ref{appenA}.
\end{enumerate}

\begin{lemma}\label{lem3}
Let $\Phi$ be a $k\times(k-1)$ Vandermonde matrix defined as in (\ref{phi}), and $\Delta$ be a $k\times k$ diagonal matrix with distinct and nonzero diagonal elements. Suppose $$
X=\Phi S+\Delta\Phi T,
$$
where $S$ and $T$ are two $(k-1)\times(k-1)$ symmetric matrices. Then $S$ and $T$ can be uniquely computed from $X,\Phi$ and $\Delta$.
\end{lemma}

\section{mscr codes with $d=2k-3$ and $t=2$}\label{sec3}
As a warm-up, we first construct a scalar MSCR code for the special case of $t=2$ and $d=2k-3$. By the parameters of MSCR codes given in (\ref{mscr}), we know in this case
$$\alpha=d-k+t=k-1,~ B=k\alpha=k(k-1)\;.$$
Since we are to give the construction using the product matrix framework, the key point is to design the encoding matrix $G$ and the message matrix $M$.

Let $\alpha_{1},\alpha_{2},...,\alpha_{n}$ be $n$ distinct nonzero elements in $\mathbb{F}_q$ such that $\alpha_{i}^{k-2}\neq \alpha_{j}^{k-2}$ for $1\leq i\neq j\leq n$. In particular, set $q-1\geq n$ and $q-1$ coprime with $k-2$, then any $n$ distinct elements in $\mathbb{F}_q^*$ can be chosen as $\alpha_i$'s. Our encoding matrix is defined as follows, which is an $n\times d$ Vandermonde matrix.
$$G =
\begin{pmatrix} 1 & \alpha_{1} &\cdots & \alpha_{1}^{2k-4}
\\ 1 & \alpha_{2} &\cdots& \alpha_{2}^{2k-4} \\
 \vdots & \vdots & \vdots & \vdots \\
  1 & \alpha_{n} &\cdots & \alpha_{n}^{2k-4} \end{pmatrix}.
$$
The message matrix $M$ is a $d\times \alpha$ matrix which has the following form.
$$ M=
\begin{pmatrix} {\displaystyle S_{\scriptscriptstyle k-1}} \\
 {\textstyle 0_{\scriptscriptstyle k-2}} \\ \end{pmatrix}+\begin{pmatrix} {\textstyle 0_{\scriptscriptstyle k-2}} \\ {\displaystyle T_{\scriptscriptstyle k-1}}
 \end{pmatrix},$$
where $S$ and $T$ are two $(k\!-\!1)\!\times\!(k\!-\!1)$ symmetric matrices each filled with $\frac{k(k-1)}{2}$ message symbols chosen from $\mathbb{F}_q$. Thus the total number of message symbols contained in $M$ is $k(k-1)=B$. The subscript $k-1$ or $k-2$ denotes the number of rows, thus here $\textstyle 0_{\scriptscriptstyle k-2}$ means a $(k\!-\!2)\!\times\!(k\!-\!1)$ all-zero matrix.

\begin{remark}\label{remark1}
Note that the symmetric matrices $S$ and $T$ respectively form the first and the last $k-2$ rows of $M$, while interweave in the $(k\!-\!1)$-th row of $M$.
This partial interweaving structure is the key idea in our design of MSCR codes. Actually, comparing with MSR codes which have $\alpha=d-k+1$, MSCR codes have $\alpha=d-k+t$ for $t\geq 2$. That is, fixing $\alpha$, the MSCR codes have $d$ less than the MSR codes by $t-1$. Since the message matrix has $d$ rows, the reduction on $d$ for MSCR codes is realized by interweaving $S$ and $T$ in $t-1$ rows. On the other hand, the interweaved message symbols just can be unpicked through the exchanging phase between $t$ newcomers in the cooperative repair.
\end{remark}

In the next, we illustrate the data reconstruction and the $t$-node cooperative repair for the construction.

\begin{theorem}(Data reconstruction)
The $B$ message symbols can be recovered from any $k$ nodes.
\end{theorem}

\begin{proof} For any $k$ nodes $i_{1}, i_{2},...,i_{k}\in[n]$, denote by $\Psi$ the sub-matrix of $G$ restricted to the rows corresponding to the $k$ nodes.
Then the data collector can obtain the symbols
\begin{equation*}
\begin{aligned}
\Psi M&=\Psi\begin{pmatrix} {\displaystyle S} \\
 {\textstyle 0} \\ \end{pmatrix}+\Psi\begin{pmatrix} {\textstyle 0} \\ {\displaystyle T}
  \\ \end{pmatrix} \\
&=\Phi S+\Delta \Phi T,
\end{aligned}
\end{equation*}
where $\Phi$ denotes the $k\times(k-1)$ matrix formed by the first $k-1$ columns of $\Psi$, and $\Delta={\rm diag}[\alpha_{i_{1}}^{k-2}, \alpha_{i_{2}}^{k-2},...,\alpha_{i_{k}}^{k-2} ]$. For the second equality above, one needs to note the sub-matrix formed by the last $k-1$ columns of $\Psi$ equals $\Delta\Phi$.

Then the theorem follows from Lemma \ref{lem3}.
\end{proof}

\begin{theorem}(Cooperative repair)
Suppose node $i_{1}$ and $i_{2}$ fail. Then two newcomers can regenerate node $i_1$ and $i_2$ respectively through a cooperative repair in two phases:

{Phase 1: } Each newcomer connects to any $d=2k-3$ surviving nodes as helper nodes and downloads one symbol from each helper node.

{Phase 2: } Two newcomers exchange one symbol with each other.
\end{theorem}

\begin{proof}
For simplicity, the two newcomers are also called node $i_1$ and $i_2$ respectively.
For $i\in[n]$, denote by $\bm{\psi}_{i}$ the $i$-th row of $G$, then the data stored in the $i$-th node is
$\bm{c}_{i}=\bm{\psi}_{i}M=(c_{i,1},...,c_{i,k-1})$. Moreover, define $\bm{\varphi}_{i}=(1, \alpha_{i}, ..., \alpha_{i}^{k-2})$, then  $$\bm{c}_{i}=\bm{\psi}_{i}M=\bm{\varphi}_{i}S+\alpha_{i}^{k-2}\bm{\varphi}_{i}T.$$ Therefore, node $i_{1}$ needs to recover $\bm{c}_{i_1}=\bm{\varphi}_{i_{1}}S+\alpha_{i_{1}}^{k-2}\bm{\varphi}_{i_{1}}T$, and node $i_{2}$ needs to recover $\bm{c}_{i_2}=\bm{\varphi}_{i_{2}}S+\alpha_{i_{2}}^{k-2}\bm{\varphi}_{i_{2}}T$.

In Phase 1, suppose the set of helper nodes connected by node $i_1$ is $R_1\subseteq[n]\setminus\{i_1,i_2\}$ with $|R_1|=2k-3$. Then for each $j\in R_1$, node $j$ sends the symbol
\begin{equation}\label{eq2}
\bm{c}_{j}\bm{\varphi}_{i_{1}}^\tau=\bm{\psi}_{j}M\bm{\varphi}_{i_{1}}^\tau
\end{equation} to node $i_1$.
Thus node $i_{1}$ obtains the symbols $\Psi_{repair}M\bm{\varphi}_{i_{1}}^\tau$, where $\Psi_{repair}$ is a $(2k-3)\times(2k-3)$ invertible matrix consisting of the rows $\bm{\psi}_{j}, j\in R_1$.
Therefore, node $i_{1}$ can derive the symbols $M\bm{\varphi}_{i_{1}}^\tau$ by multiplying the inverse of $\Psi_{repair}$.

Furthermore, we will show the symbols $M\bm{\varphi}_{i_{1}}^\tau$ can be used to derive partial data of $\bm{c}_{i_1}$. For convenience, denote
$$\begin{aligned}
M&=\begin{pmatrix} {\displaystyle S} \\
 {\textstyle 0} \\ \end{pmatrix}+\begin{pmatrix} {\textstyle 0} \\ {\displaystyle T}
  \\ \end{pmatrix}
 &=\begin{pmatrix} \bm{u}_{1} \\ \vdots \\ \bm{u}_{k-2}  \\ \bm{u}_{k-1} \\ \bm{0}\\ \vdots\\ \bm{0}\\ \end{pmatrix}+\begin{pmatrix} \bm{0} \\ \vdots \\ \bm{0}  \\ \bm{v}_{1} \\ \bm{v}_{2}\\ \vdots\\ \bm{v}_{k-1}\\ \end{pmatrix}
 &=\begin{pmatrix} \bm{u}_{1} \\ \vdots \\ \bm{u}_{k-2}  \\ \bm{u}_{k-1}+\bm{v}_{1} \\ \bm{v}_{2}\\ \vdots\\ \bm{v}_{k-1}\\ \end{pmatrix},
\end{aligned}$$
where $\bm{u}_{i}$ and $\bm{v}_{i}$ are the $i$-th row of $S$ and $T$ respectively, $i\in[k-1]$. Due to the symmetry of $S$ and $T$, $\bm{u}_{i}^\tau$ and $\bm{v}_{i}^\tau$ are also the $i$-th column of $S$ and $T$ respectively. Therefore,
\begin{equation}\label{eqw}
 M\bm{\varphi}_{i_{1}}^\tau\!=\!\begin{pmatrix} \bm{u}_{1}\cdot\bm{\varphi}_{i_{1}}^\tau \\ \vdots \\ \bm{u}_{k-2}\cdot\bm{\varphi}_{i_{1}}^\tau  \\ (\bm{u}_{k-1}+\bm{v}_{1})\cdot\bm{\varphi}_{i_{1}}^\tau \\ \bm{v}_{2}\cdot\bm{\varphi}_{i_{1}}^\tau\\ \vdots\\ \bm{v}_{k-1}\cdot\bm{\varphi}_{i_{1}}^\tau\\ \end{pmatrix}
 \!=\!\begin{pmatrix} \bm{\bm{\varphi}}_{i_{1}}\cdot \bm{u}_{1}^\tau \\ \vdots \\ \bm{\bm{\varphi}}_{i_{1}}\cdot \bm{u}_{k-2}^\tau  \\ \bm{\varphi}_{i_{1}}\cdot (\bm{u}_{k-1}^\tau+\bm{v}_{1}^\tau) \\ \bm{\varphi}_{i_{1}}\cdot \bm{v}_{2}^\tau\\ \vdots\\ \bm{\varphi}_{i_{1}}\cdot \bm{v}_{k-1}^\tau\\ \end{pmatrix}\!=\!\begin{pmatrix}{\omega}_1\\\vdots\\{\omega}_{k-2}\\{\omega}_{k-1}\\{\omega}_{k}\\\vdots\\{\omega}_{2k-3}\end{pmatrix},
\end{equation}
where the second equality comes from the transposition, and ${\omega}_{i}$ denotes the $i$-th coordinate of $M\bm{\varphi}_{i_{1}}^\tau$ for $i\in[2k-3]$. Recall that node $i_1$ is to recover $\bm{c}_{i_1}=(c_{i_{1},1},...,c_{i_{1},k-1})=\bm{\varphi}_{i_{1}}S+\alpha_{i_{1}}^{k-2}\bm{\varphi}_{i_{1}}T$. That is, for $j\in[k-1]$, $c_{i_{1},j}=\bm{\varphi}_{i_{1}}\bm{u}_{j}^\tau+\alpha_{i_{1}}^{k-2}\bm{\varphi}_{i_{1}}\bm{v}_{j}^\tau$. Actually, after obtaining $M\bm{\varphi}_{i_{1}}^\tau=(\omega_1,...,\omega_{2k-3})^\tau$ node $i_1$ computes the following
$$\begin{cases}
 \omega_{1}+\alpha_{i_{1}}^{k-2}\omega_{k-1}+\alpha_{i_{1}}^{2(k-2)}\omega_{2k-3}, \\
 \omega_{2}+\alpha_{i_{1}}^{k-2}\omega_{k},\\
 \omega_{3}+\alpha_{i_{1}}^{k-2}\omega_{k+1},\\
 ~~~~~\vdots \\
 \omega_{k-2}+\alpha_{i_{1}}^{k-2}\omega_{2k-4}.
\end{cases} $$
It can be verified that
\begin{eqnarray*}
&&\omega_{1}+\alpha_{i_{1}}^{k-2}\omega_{k-1}+\alpha_{i_{1}}^{2(k-2)}\omega_{2k-3}\\
&=&\bm{\varphi}_{i_{1}}\cdot \bm{u}_{1}^\tau+\alpha_{i_{1}}^{k-2}\bm{\varphi}_{i_{1}}\cdot (\bm{u}_{k-1}^\tau+\bm{v}_{1}^\tau)+\alpha_{i_{1}}^{2(k-2)}\bm{\varphi}_{i_{1}}\cdot \bm{v}_{k-1}^\tau \\
&=&\bm{\varphi}_{i_{1}}\cdot(\bm{u}_{1}^\tau+\alpha_{i_{1}}^{k-2}\bm{v}_{1}^\tau)+\alpha_{i_{1}}^{k-2}\bm{\varphi}_{i_{1}}\cdot(\bm{u}_{k-1}^\tau+\alpha_{i_{1}}^{k-2}\bm{v}_{k-1}^\tau) \\
&=&c_{i_{1},1}+\alpha_{i_{1}}^{k-2}c_{i_{1},k-1}.
\end{eqnarray*}
while for $j\in\{2,...,k-2\}$,
$$
\omega_{j}+\alpha_{i_{1}}^{k-2}\omega_{j+k-2}=\bm{\varphi}_{i_{1}}\cdot \bm{u}_{j}^\tau+\alpha_{i_{1}}^{k-2}\bm{\varphi}_{i_{1}}\cdot \bm{v}_{j}^\tau\\
=c_{i_{1},j}\;.$$
That is, after Phase 1 node $i_1$ gets $d$ symbols $M\bm{\varphi}_{i_{1}}^\tau$, from which it further recovers $c_{i_1,2},...,c_{i_1,k-2}$ and $c_{i_{1},1}+\alpha_{i_{1}}^{k-2}c_{i_{1},k-1}$.

Similarly, by connecting to $d$ helper nodes and downloading $\bm{c}_j\bm{\varphi}_{i_2}^\tau$ from each helper node $j$, node $i_2$ can get $d$ symbols $M\bm{\varphi}_{i_{2}}^\tau$, from which it can further recover $c_{i_2,2},...,c_{i_2,k-2}$ and $c_{i_2,1}+\alpha_{i_2}^{k-2}c_{i_{2},k-1}$.

\vspace{4pt}

In Phase 2, node $i_1$ sends the symbol $\bm{\psi}_{i_{2}}M\bm{\varphi}_{i_{1}}^\tau$ to $i_{2}$, and node $i_2$ sends the symbol $\bm{\psi}_{i_{1}}M\bm{\varphi}_{i_{2}}^\tau$ to $i_{1}$.
So, node $i_1$ gets
$$
\bm{\psi}_{i_{1}}M\bm{\varphi}_{i_{2}}^\tau=\bm{c}_{i_{1}}\bm{\varphi}_{i_{2}}^\tau=(c_{i_{1},1},...,c_{i_{1},k-1})\begin{pmatrix}1\\\alpha_{i_{2}}\\ \vdots \\ \alpha_{i_{2}}^{k-2} \end{pmatrix},
$$
from which node $i_1$ can obtain the value of $c_{i_{1},1}+\alpha_{i_{2}}^{k-2}c_{i_{1},k-1}$ because node $i_1$ has already obtained $c_{i_{1},j}$ for $ j\in\{2,..., k-2\}$ after Phase 1.
Thus $c_{i_{1},1}$ and $c_{i_{1},k-1}$ can be solved from
\[\begin{cases}
 c_{i_{1},1}+\alpha_{i_{1}}^{k-2}c_{i_{1},k-1} \\
 c_{i_{1},1}+\alpha_{i_{2}}^{k-2}c_{i_{1},k-1}
\end{cases} \]
since $\alpha_{i_{1}}^{k-2}\neq\alpha_{i_{2}}^{k-2}$.

In the same way, node $i_{2}$ can further recover $c_{i_{2},1}$ and  $c_{i_{2},k-1}$ after Phase 2.

\end{proof}

\section{MSCR codes with $d=2k-1-t$ and $2\leq t\leq k-1$}\label{sec4}
In this section, the construction for $t=2$ is extended to that for more general $t$, i.e. $2\leq t\leq k-1$, while $d$ is restricted to $d=2k-1-t$. The restrictions on the parameters are explained in Remark \ref{remark2}.  Obviously, in this case we still have $\alpha=d-k+t=k-1$ and $B=k\alpha=k(k-1)$.

We first define the $n\times d$ encoding matrix $G$ and the $d\times\alpha$ message matrix $M$. For simplicity, denote $\mu=k-t$. Let $\alpha_{1},...,\alpha_{n}$ be $n$ distinct nonzero elements in $\mathbb{F}_q$ such that $\alpha_{i}^{\mu}\neq \alpha_{j}^{\mu}$ for $1\leq i\neq j\leq n$. Then, define
$$G=\begin{pmatrix}
1 & \alpha_{1} &\cdots & \alpha_{1}^{2k-2-t} \\
1 & \alpha_{2} &\cdots & \alpha_{2}^{2k-2-t} \\
 \vdots & \vdots & \vdots & \vdots \\
1 & \alpha_{n} &\cdots
 & \alpha_{n}^{2k-2-t}
\end{pmatrix},~~M=
\begin{pmatrix} {\displaystyle S_{\scriptscriptstyle k-1}} \\
 {\textstyle 0_{\scriptscriptstyle \mu}} \\ \end{pmatrix}+\begin{pmatrix} {\textstyle 0_{\scriptscriptstyle \mu}} \\ {\displaystyle T_{\scriptscriptstyle k-1}}
 \end{pmatrix},$$
where $S$ and $T$ are two $(k-1)\times(k-1)$ symmetric matrices each filled with $\frac{k(k-1)}{2}$ message symbols chosen from $\mathbb{F}_q$. Thus the total number of message symbols contained in $M$ is $k(k-1)=B$. Note that $k-1+\mu=k-1+k-t=d$, thus $M$ is a $d\times\alpha$ matrix.

The data reconstruction can be verified as before. Specifically, suppose the data collector connects to $k$ nodes $i_{1},...,i_{k}\in[n]$. Let $\Psi$ be the $k\times d$ sub-matrix of $G$ consisting of the rows $\bm{\psi}_{i_j}, j\in[k]$, then  the data collector has access to the symbols
\begin{equation*}
\begin{aligned}
\Psi M&=\Psi\begin{pmatrix} {\displaystyle S} \\
 {\textstyle 0} \\ \end{pmatrix}+\Psi\begin{pmatrix} {\textstyle 0} \\ {\displaystyle T}
  \\ \end{pmatrix} \\
&=\Phi S+\Delta' \Phi T,
\end{aligned}
\end{equation*}
where $\Phi$ is a $k\times(k-1)$ matrix consisting of the first $k-1$ columns of $\Psi$ and $\Delta'$ is the $k\times k$ diagonal matrix ${\rm diag}[\alpha_{i_{1}}^{\mu}, \alpha_{i_{2}}^{\mu},...,\alpha_{i_{k}}^{\mu} ]$. The diagonal elements of $\Delta'$ are all distinct and nonzero by our construction. Thus the matrix $S$ and $T$ can be reconstructed from the $k$ nodes by Lemma \ref{lem3}.

\begin{remark}\label{remark2}
The reason $d$ is restricted to $d=2k-1-t$ is that we want to keep $\alpha=d-k+t=k-1$ unchanged so that Lemma \ref{lem3} still can be used to ensure the data reconstruction. On the other hand, in this section we have $2\leq t\leq k-1$ rather than the most general condition $2\leq t\leq n-k$ \footnote{The condition $t\leq n-k$ is necessary to ensure the surviving nodes (i.e., $n-t\geq k$ nodes) can recover the original data file.}, because $t\leq k-1$ implies $t-1<k-1$ which means the two symmetric matrices $S$ and $T$ are not totally interweaved and thus the diagonal matrix $\Delta'$ has $k$ distinct diagonal elements as required by Lemma \ref{lem3}.
\end{remark}

Next we mainly illustrate the cooperative repair of $t$ nodes.

\begin{theorem}(Cooperative repair)
Without loss of generality, suppose the set of failed nodes is $\{1,2,...,t\}$.  Then $t$ newcomers can recover the $t$ failed nodes respectively by each connecting to $d$ surviving nodes as helpers nodes and downloading one symbol from each of the helper nodes in Phase 1, and then downloading one symbol from each of the other $t-1$ newcomers in Phase 2. Therefore, the repair bandwidth for recovering every failed node is $d+t-1$.
\end{theorem}
\begin{proof}
For simplicity, the $t$ newcomers are also called node $1,...,t$ respectively in the following. Using the notations defined before, the data stored in node $i$ is \begin{equation}\label{eqc}\bm{c}_{i}=\bm{\psi}_{i}M=\bm{\bm{\varphi}}_{i}S+\alpha_{i}^{\mu}\bm{\varphi}_{i}T=(c_{i,1},...,c_{i,k-1})\;.\end{equation}
That is, node $i$ is to recover $\bm{c}_{i}$ for $i\in[t]$.

In Phase 1,  fix $i\in[t]$, and let $R_i\subseteq[n]\setminus[t]$ with $|R_i|=d$ be the set of helper nodes connected by node $i$. Then
for each $j\in R_i$, node $j$ sends node $i$ the following symbol
\begin{equation}\label{eq4}
c_{j}\bm{\varphi}_{i}^\tau=\bm{\psi}_{j}M\bm{\varphi}_{i}^\tau\;.
\end{equation}
Thus node $i$ receives the symbols $\{\bm{\psi}_{j}M\bm{\varphi}_{i}^\tau|j\in R_i\}$. Since all the $\bm{\psi}_{j}$'s, $j\in R_i$, form a $d\times d$ invertible matrix, then node $i$ can obtain the symbols $M\bm{\varphi}_{i}^\tau$.

Similar to (\ref{eqw}), it has
$$
 M\bm{\varphi}_{i}^\tau=\begin{pmatrix}
  \bm{u}_{1}\cdot\bm{\varphi}_{i}^\tau \\ \vdots \\ \bm{u}_{\mu}\cdot\bm{\varphi}_{i}^\tau  \\ (\bm{u}_{\mu+1}+\bm{v}_{1})\cdot\bm{\varphi}_{i}^\tau \\ \vdots \\ (\bm{u}_{k-1}+\bm{v}_{t-1})\cdot\bm{\varphi}_{i}^\tau \\ \bm{v}_{t}\cdot\bm{\varphi}_{i}^\tau\\ \vdots\\ \bm{v}_{k-1}\cdot\bm{\varphi}_{i}^\tau\\
 \end{pmatrix}
 =\begin{pmatrix}
  \bm{\varphi}_{i}\cdot \bm{u}_{1}^\tau \\ \vdots \\ \bm{\varphi}_{i}\cdot \bm{u}_{\mu}^\tau  \\ \bm{\varphi}_{i}\cdot (\bm{u}_{\mu+1}^\tau+\bm{v}_{1}^\tau) \\ \vdots \\ \bm{\varphi}_{i}\cdot (\bm{u}_{k-1}^\tau+\bm{v}_{t-1}^\tau) \\ \bm{\varphi}_{i}\cdot \bm{v}_{t}^\tau\\ \vdots\\
  \bm{\varphi}_{i}\cdot \bm{v}_{k-1}^\tau\end{pmatrix}
 =\begin{pmatrix}
 \omega_{1} \\ \vdots \\ \omega_{\mu}  \\ \omega_{\mu+1} \\ \vdots \\ \omega_{k-1} \\ \omega_{k}\\ \vdots\\ \omega_{2k-1-t}\\
 \end{pmatrix},
 $$
where the second equality comes from the transposition. Therefore, for each $l\in[2k-1-t]$,
 $$\omega_{l}= \left\{\begin{array}{ll}
 \bm{\varphi}_{i}\bm{u}_{l}^\tau,~~&  1\leq l\leq \mu \\
 \bm{\varphi}_{i}(\bm{u}_{l}^\tau+\bm{v}_{l-\mu}^\tau),~~&  \mu+1\leq l\leq k-1 \\
 \bm{\varphi}_{i}\bm{v}_{l-\mu}^\tau,~~&  k\leq l\leq 2k-1-t.
 \end{array}\right.$$
Then we will give some calculations on the symbols $\omega_1,...,\omega_{2k-1-t}$ which help to recover $\bm{c}_i$.
In particular,  denote $d=2k-1-t=z\mu+r$ for some integers $z$ and $r$ where $0\leq r\leq \mu-1$. Then $\alpha=k-1=d-\mu=(z-1)\mu+r$.
Next node $i$  computes the following $\mu$ symbols
\begin{equation}\label{eqeq7}\begin{cases}
 \omega_{1}+\alpha_{i}^{\mu}\omega_{\mu+1}+\alpha_{i}^{2\mu}\omega_{2\mu+1}+\cdots+\alpha_{i}^{(z-1)\mu}\omega_{(z-1)\mu+1}+\alpha_{i}^{z\mu}\omega_{z\mu+1}, \\
 ~~~~~\vdots \\
 \omega_{r}+\alpha_{i}^{\mu}\omega_{\mu+r}+\alpha_{i}^{2\mu}\omega_{2\mu+r}+\cdots+\alpha_{i}^{(z-1)\mu}\omega_{(z-1)\mu+r}+\alpha_{i}^{z\mu}\omega_{z\mu+r},\\
 \omega_{r+1}+\alpha_{i}^{\mu}\omega_{\mu+r+1}+\alpha_{i}^{2\mu}\omega_{2\mu+r+1}+\cdots+\alpha_{i}^{(z-1)\mu}\omega_{(z-1)\mu+r+1},\\
 ~~~~~\vdots \\
 \omega_{\mu}+\alpha_{i}^{\mu}\omega_{2\mu}+\alpha_{i}^{2\mu}\omega_{3\mu}+\cdots+\alpha_{i}^{(z-1)\mu}\omega_{z\mu}.
\end{cases} \end{equation}
Recall that from (\ref{eqc}) it has  $c_{i,j}=\bm{\varphi}_{i}\bm{u}_{j}^\tau+\alpha_{i}^{\mu}\bm{\varphi}_{i}\bm{v}_{j}^\tau$ for $j\in[k-1]$.
Then, it can be verified that for $1\leq l\leq r$,
\begin{eqnarray*}
&&\omega_{l}+\alpha_{i}^{\mu}\omega_{\mu+l}+\cdots+\alpha_{i}^{(z-1)\mu}\omega_{(z-1)\mu+l}+\alpha_{i}^{z\mu}\omega_{z\mu+l}\\
&=&\bm{\varphi}_{i}\Big(\bm{u}_{l}^\tau+\alpha_{i}^{\mu}(\bm{u}_{\mu+l}^\tau+\bm{v}_{l}^\tau)+\cdots+\alpha_{i}^{(z-1)\mu}(\bm{u}_{(z-1)\mu+l}^\tau+\bm{v}_{(z-2)\mu+l}^\tau)+\alpha_{i}^{z\mu}\bm{v}_{(z-1)\mu+l}^\tau\Big)\\
&=&\bm{\varphi}_{i}\Big(\big(\bm{u}_{l}^\tau+\alpha_{i}^{\mu}\bm{v}_{l}^\tau\big)+\alpha_{i}^{\mu}\big(\bm{u}_{\mu+l}^\tau+\alpha_{i}^{\mu}\bm{v}_{\mu+l}^\tau\big)+\cdots+\alpha_{i}^{(z-1)\mu}\big(\bm{u}_{(z-1)\mu+l}^\tau+\alpha_{i}^{\mu}\bm{v}_{(z-1)\mu+l}^\tau\big)\Big)\\
&=&c_{i,l}+\alpha_{i}^{\mu}c_{i,\mu+l}+\cdots+\alpha_{i}^{(z-1)\mu}c_{i,(z-1)\mu+l},
\end{eqnarray*}
while for $r+1\leq l\leq \mu$,
\begin{eqnarray*}
&&\omega_{l}+\alpha_{i}^{\mu}\omega_{\mu+l}+\cdots+\alpha_{i}^{(z-2)\mu}\omega_{(z-2)\mu+l}+\alpha_{i}^{(z-1)\mu}\omega_{(z-1)\mu+l}\\
&\stackrel{\tiny a}{=}&\bm{\varphi}_{i}\Big(\bm{u}_{l}^\tau+\alpha_{i}^{\mu}(\bm{u}_{\mu+l}^\tau+\bm{v}_{l}^\tau)+\cdots+\alpha_{i}^{(z-2)\mu}(\bm{u}_{(z-2)\mu+l}^\tau+\bm{v}_{(z-3)\mu+l}^\tau)+\alpha_{i}^{(z-1)\mu}\bm{v}_{(z-2)\mu+l}^\tau\Big)\\
&=&\bm{\varphi}_{i}\Big(\big(\bm{u}_{l}^\tau+\alpha_{i}^{\mu}\bm{v}_{l}^\tau\big)+\alpha_{i}^{\mu}\big(\bm{u}_{\mu+l}^\tau+\alpha_{i}^{\mu}\bm{v}_{\mu+l}^\tau\big)+\cdots+\alpha_{i}^{(z-2)\mu}\big(\bm{u}_{(z-2)\mu+l}^\tau+\alpha_{i}^{\mu}\bm{v}_{(z-2)\mu+l}^\tau\big)\Big)\\
&=&c_{i,l}+\alpha_{i}^{\mu}c_{i,\mu+l}+\cdots+\alpha_{i}^{(z-2)\mu}c_{i,(z-2)\mu+l},
\end{eqnarray*}
where the equality $\stackrel{\tiny a}{=}$ is due to the fact that when $r+1\leq l\leq \mu$, $\omega_{(z-1)\mu+l}=\bm{\varphi}_i \bm{v}_{(z-2)\mu+l}^\tau$. Therefore, the calculations in (\ref{eqeq7}) actually give $\mu$ linear equations on $c_{i,1}, c_{i,2},...,c_{i,k-1}$. Furthermore, write the $\mu$ linear equations in the matrix form $H_{i,1}\bm{c}_{i}^\tau$, then the coefficient matrix $H_{i,1}$ has the form :
\begin{equation*}
H_{i,1}=\begin{pmatrix}I_{\mu}&\alpha_{i}^{\mu}I_{\mu}&\cdots&\alpha_{i}^{(z-2)\mu}I_{\mu}&\alpha_{i}^{(z-1)\mu}I_{\mu}^{(r)}\end{pmatrix},
\end{equation*}
where $I_{\mu}$ denotes the $\mu \times \mu$ identity matrix, and  $I_{\mu}^{(r)}$ denotes the $\mu\times r$ matrix consisting of the first $r$ columns of $I_{\mu}$.

\vspace{4pt}

Now we turn to Phase 2 and still fix $i\in[t]$. Since each node $j\in [t]\setminus\{i\}$ has known $M\bm{\varphi}_{j}^\tau$ after Phase 1, then node $j$  sends the symbol $\bm{\psi}_{i}M\bm{\bm{\varphi}}_{j}^\tau=\bm{c}_{i}\bm{\bm{\varphi}}_{j}^\tau$ to node $i$. Thus in Phase 2 node $i$  receives $t-1$ more symbols $
\{\bm{c}_{i}\bm{\varphi}_{j}^\tau|\ j\in[t],~j\neq i\}
$
which correspond to $t-1$ linear equations on $c_{i,1},c_{i,2},..., c_{i,k-1}$. Write these $t-1$ linear equations in matrix form $H_{i,2}\bm{c}_{i}^\tau$, then the coefficient matrix $H_{i,2}$ has the form
\begin{equation*}
H_{i,2}=\left(
\begin{array}{ccccc}
1&\alpha_{1}&\alpha_{1}^{2}&\cdots\cdots& \alpha_{1}^{k-2}\\
\vdots&\vdots&\vdots&\vdots&\vdots \\
1&\alpha_{i-1}&\alpha_{i-1}^{2}&\cdots\cdots& \alpha_{i-1}^{k-2}\\
1&\alpha_{i+1}&\alpha_{i+1}^{2}&\cdots\cdots& \alpha_{i+1}^{k-2}\\
\vdots&\vdots&\vdots&\vdots& \vdots\\
1&\alpha_{t}&\alpha_{t}^{2}&\cdots\cdots&\alpha_{t}^{k-2}\\
\end{array}
\right)\;.
\end{equation*}

Therefore, after the two phases, node $i$ obtains $\mu+t-1=k-t+t-1=k-1$ linear equations on $c_{i,1},c_{i,2},..., c_{i,k-1}$. Moreover, this linear system has coefficient matrix
\begin{equation}\label{eq8}
H=\begin{pmatrix} H_{i,1}\\H_{i,2}\end{pmatrix}=
\begin{pmatrix}
I_{\mu} & \alpha_{i}^{\mu}I_{\mu}&\alpha_{i}^{2\mu}I_{\mu} &...&\alpha_{i}^{(z-2)\mu}I_{\mu}&\alpha_{i}^{(z-1)\mu}I_{\mu}^{(r)} \\
P & \widetilde{\Delta} P &\widetilde{\Delta}^{2} P&\cdots&\widetilde{\Delta}^{z-2} P&\widetilde{\Delta}^{z-1} P^{(r)}
\end{pmatrix},
\end{equation}
where $P$ denotes the $(t-1)\times\mu$ matrix consisting of the first $\mu$ columns of $H_{i,2}$, $P^{(r)}$ is the $(t-1)\times r$ sub-matrix of $P$ restricted to the first $r$ columns, and $\widetilde{\Delta}={\rm diag}[\alpha_{1}^{\mu},...,\alpha_{i-1}^{\mu},\alpha_{i+1}^{\mu},...,\alpha_{t}^{\mu}]$.

Finally, we will show that the coefficient matrix $H$ in (\ref{eq8}) is invertible, thus node $i$ can solve $c_{i,1},c_{i,2},..., c_{i,k-1}$.
Note that the $k-1$ columns of $H$ are divided into $z$ column blocks where the first $z-1$ blocks each has $\mu$ columns while the last block has $r$ columns. Then executing elementary column transformations on $H$, i.e., multiplying the first $r$ columns of the $(z-1)$-th block by $\alpha_{i}^{\mu}$ , and subtracting from the $z$-th block, we can get
\begin{equation*}
H'=
\begin{pmatrix}
I_{\mu} & \alpha_{i}^{\mu}I_{\mu}&\alpha_{i}^{2\mu}I_{\mu} &\cdots&\alpha_{i}^{(z-2)\mu}I_{\mu}&\emph{0} \\
P & \widetilde{\Delta} P &\widetilde{\Delta}^{2} P&\cdots&\widetilde{\Delta}^{z-2} P&(\widetilde{\Delta}-\alpha_{i}^{\mu}I_{t-1})\widetilde{\Delta}^{z-2} P^{(r)}
\end{pmatrix}.
\end{equation*}

For $j=z-2,z-3,...,1$, multiply the $j$-th block by $\alpha_{i}^{\mu}$ , and subtract from the $(j+1)$-th block, then we finally get
\begin{equation*}
\begin{aligned}
H''
=&\begin{pmatrix}
I_{\mu} & \emph{0}&\emph{0} &\cdots&\emph{0}&\emph{0} \\
P & (\widetilde{\Delta}-\alpha_{i}^{\mu}I_{t-1}) P &(\widetilde{\Delta}-\alpha_{i}^{\mu}I_{t-1})\widetilde{\Delta} P&\cdots&(\widetilde{\Delta}-\alpha_{i}^{\mu}I_{t-1})\widetilde{\Delta}^{z-3} P&(\widetilde{\Delta}-\alpha_{i}^{\mu}I_{t-1})\widetilde{\Delta}^{z-2} P^{(r)}
\end{pmatrix}     \\
=&\begin{pmatrix}
I_{\mu} & {\emph{0}_{\mu\times(t-1)}} \\
P & (\widetilde{\Delta}-\alpha_{i}^{\mu}I_{t-1})\widetilde{P}
\end{pmatrix},
\end{aligned}
\end{equation*}
where
$
\widetilde{P}=(P \ \ \ \widetilde{\Delta} P \ \ \ \cdots \ \ \ \widetilde{\Delta}^{z-3} P \ \ \ \widetilde{\Delta}^{z-2}P^{(r)})
$
is a $(t-1)\times(t-1)$ Vandermonde matrix and thus is invertible.
Note that $\widetilde{\Delta}-\alpha_{i}^{\mu}I_{t-1}$ is a $(t-1)\times(t-1)$ diagonal matrix ${\rm diag}[\alpha_{1}^{\mu}-\alpha_{i}^{\mu},...,\alpha_{i-1}^{\mu}-\alpha_{i}^{\mu},\alpha_{i+1}^{\mu}-\alpha_{i}^{\mu},...,\alpha_{t}^{\mu}-\alpha_{i}^{\mu}]$. Since $\alpha^\mu_i\neq \alpha^\mu_j$ for $1\leq i\neq j\leq n$,
thus $\widetilde{\Delta}-\alpha_{i}^{\mu}I_{t-1}$ is also invertible. As a result, the matrix $H''$ is invertible. Since only elementary column transformations are executed from $H$ to $H''$, it immediately follows that $H$ is an invertible matrix.

For fixed $i\in[t]$, we have shown that node $i$ can recover ${\bm c}_i$ as required by the cooperative repair property. The repair of other failed nodes goes in the same way and the theorem is proved.
\end{proof}

\section{MSCR codes with $d\geq \max \{2k-1-t,k\}$ and $2\leq t\leq n-k$}\label{sec5}
In this section, we first show that any MSCR code can be transformed into a systematic MSCR code with the same parameters. In particular, the codes constructed in Section \ref{sec4} could have a systematic form. Then by applying a shortening technique to these codes, we build scalar MSCR codes for all $2\leq t\leq n-k$ and $d\geq \max \{2k-1-t,k\}$.

\subsection{Systematic MSCR codes}\label{sec5a}
Since $B=k\alpha$ for MSCR codes, the original data file can be denoted as $(\bm{m}_1,...,\bm{m}_k)$ where each $\bm{m}_i, i\in[k]$, consists of $\alpha$ symbols. An MSCR code is called systematic if there exist $k$ nodes (called systematic nodes) which store $\bm{m}_i, i\in[k]$, respectively. In fact, through a reverse application of the data reconstruction, we can turn any MSCR code into a systematic one.

\begin{theorem}\label{thmthm5}
Suppose there exists an MSCR code $\mathcal{C}$ with parameters $(n,k,d,t,\alpha,\beta)$, then for any $I\subseteq[n]$ with $|I|=k$, there exists an $(n,k,d,t,\alpha,\beta)$ systematic MSCR code $\mathcal{C}'$ taking the nodes in $I$ as systematic nodes.
\end{theorem}
\begin{proof}
In general, we define the MSCR code $\mathcal{C}$  by using its encoding map $\mathcal{E}:~F^{B}\rightarrow (F^\alpha)^n$, $\mathcal{E}(\bm{m})=(\bm{c}_1,...,\bm{c}_n)$, that is, for a data file $\bm{m}$ of size $B$, node $i$ stores $\bm{c}_i$ for $i\in[n]$.

From the data reconstruction property, the content stored in any $k$ nodes uniquely determines the data file. In particular, for $I\subseteq[n]$ with $|I|=k$, there exists a reconstruction function $\mathcal{R}_{I}:~(F^\alpha)^k\rightarrow F^{B}$ such that
\begin{equation}\label{eqeq9}\mathcal{R}_{I}\big(\mathcal{E}|_I(\bm{m})\big)=\bm{m},~~\forall \bm{m}\in F^B\;, \end{equation}
where $\mathcal{E}|_I(\bm{m})$ denotes $\mathcal{E}(\bm{m})$ restricted to the nodes in $I$. Moreover, (\ref{eqeq9}) implies that both $\mathcal{R}_{I}$ and $\mathcal{E}|_I$ are one-to-one maps and $\mathcal{E}|_I=\mathcal{R}_{I}^{-1}$.

For any data file $\bm{m}\in F^B=F^{k\alpha}$, denote  $\bm{m}=(\bm{m}_1,...,\bm{m}_k)$ where $\bm{m}_i\in F^\alpha$ for $i\in[k]$. Then we define a new MSCR code by the encoding function $\mathcal{E}':~F^{B}\rightarrow (F^\alpha)^n$, such that
$$\mathcal{E}'(\bm{m})=\mathcal{E}(\mathcal{R}_{I}(\bm{m}_1,...,\bm{m}_k))\;.$$
Then we will say that $\mathcal{E}'$ actually defines the systematic MSCR code $\mathcal{C}'$ as required by the theorem.

First, because $\mathcal{R}_{I}$ and $\mathcal{E}|_{I'}$ are one-to-one maps, it follows that $\mathcal{E}'|_{I'}$  is a one-to-one map for any $I'\subseteq[n]$ with $|I'|=k$. Define the inverse map of $\mathcal{E}'|_{I'}$ as the reconstruction function for $I'$, then $\mathcal{C}'$ satisfies the data reconstruction property. Moreover,
\begin{eqnarray*}\mathcal{E}'|_I(\bm{m})&=&\mathcal{E}|_I(\mathcal{R}_{I}(\bm{m}_1,...,\bm{m}_k))\\
&=&\mathcal{R}_{I}^{-1}(\mathcal{R}_{I}(\bm{m}_1,...,\bm{m}_k))\\
&=&(\bm{m}_1,...,\bm{m}_k)\;,\end{eqnarray*}
so the $k$ nodes in $I$ are systematic nodes. For the repair property, since $\mathcal{C}'$ and $\mathcal{C}$, the original MSCR code,  have the same codeword space except that they have different encoding maps, the repair property of $\mathcal{C}$ are maintained in $\mathcal{C}'$. It is easy to verify that $\mathcal{C}'$ and $\mathcal{C}$ have the same parameters. The theorem is proved.
\end{proof}

In particular, for the $(n,k,d,t)$ scalar MSCR code constructed under the product matrix framework, if we want the nodes $\{1,...,k\}$ to be systematic nodes, then for any data file $\bm{m}=(\bm{m}_1,...,\bm{m}_k)\in (F^\alpha)^k$ we first solve the message matrix $M(\bm{m})$ from $\Psi M(\bm{m})=W(\bm{m})$ through the data reconstruction process, where $\Psi$ denotes the encoding matrix $G$ restricted to the first $k$ rows, and $W(\bm{m})$ is the $k\times\alpha$ matrix whose $k$ rows are exactly $\bm{m}_1,...,\bm{m}_k$. Thus we obtain a systematic $(n,k,d,t)$ scalar MSCR code which encodes the data file $\bm{m}$  as $C=GM(\bm{m})$.

\subsection{Scalar MSCR codes with $d\geq \max \{2k-1-t,k\}$ and $2\leq t\leq n-k$}\label{sec5b}
In Section \ref{sec4}, we have constructed scalar MSCR codes for any $2\leq t\leq k-1$ and $d=2k-1-t$. Next we show
by proper shortening from these codes, one can derive scalar MSCR codes for all $d\geq \max \{2k-1-t,k\}$ and $2\leq t\leq n-k$. First, Theorem \ref{thm8} states the relations between the parameters of the original MSCR code and the shortened code.
The shortening technique is specifically described in the proof of Theorem \ref{thm8}. Then in Corollary \ref{coro7}, applying the shortening technique to a previously constructed MSCR code, it gives the code we want in this Section.

\begin{theorem}\label{thm8}
 If there exists an $(n^{\prime}=n+\delta,k^{\prime}=k+\delta,d^{\prime}=d+\delta,t)$ scalar MSCR code $\mathcal{C}'$ for some $\delta\geq0$, then there must exist an $(n,k,d,t)$ scalar MSCR code $\mathcal{C}$.
\end{theorem}
\begin{proof}
By Theorem \ref{thmthm5}, we can assume that $\mathcal{C}'$ is an $(n',k',d',t)$ systematic scalar MSCR code with systematic nodes $1,...,k'$. From (\ref{mscr}) we know that
\begin{eqnarray*}
 \alpha^{\prime}&=&d^{\prime}-k^{\prime}+t \\
 &=&d-k+t
\end{eqnarray*}
  and
\begin{eqnarray*}
B'&=&k^{\prime}\alpha^{\prime} \\
 &=&(k+\delta)(d-k+t)\;,
\end{eqnarray*}
while a scalar MSCR code with parameters $(n,k,d,t)$ has
$$
\alpha=d-k+t,~~B=k\alpha\;.
 $$
Thus it has
$$\alpha'=\alpha,~~ B'=B+\delta\alpha\;.$$
Now consider all the codewords in $\mathcal{C}'$ that have zeros in the first $\delta$ nodes and then puncture these codewords in the first $\delta$ nodes, it gives the desired $(n,k,d,t)$ MSCR code $\mathcal{C}$.

More specifically, in the data reconstruction any $k$ nodes in $\mathcal{C}$ plus $\delta$ imaginary systematic nodes that store all zeros correspond to $k'$ nodes in $\mathcal{C}'$ which uniquely determines a data file of length $B'$ with the first $\delta\alpha$ symbols being zeros, therefore any $k$ nodes in $\mathcal{C}$ uniquely determines a data file of size $B'-k\delta=B$. The cooperative repair of any $t$ nodes in $\mathcal{C}$ with each connecting to $d$ helper nodes can be done as the cooperative repair of the $t$ nodes in $\mathcal{C}'$ with each connecting to the $d$ helper nodes and $\delta$ imaginary nodes that store all zeros.  Therefore, one can see that $\mathcal{C}$ is an $(n,k,d,t)$ scalar MSCR code.
\end{proof}

\begin{corollary}\label{coro7}
 For any $2\leq t\leq n-k$ and $d\geq \max \{2k-1-t,k\}$, there exists an $(n,k,d,t)$ scalar MSCR code.
\end{corollary}
\begin{proof}
Define $\delta=d-(2k-1-t)\geq 0$, and let $n^{\prime}=n+\delta, k^{\prime}=k+\delta, d^{\prime}=d+\delta$. It is easy to verify that $d^{\prime}=2k^{\prime}-1-t$. Since $t\leq d-k+t=k^{\prime}-1$, then we can obtain an $(n^{\prime},k^{\prime},d^{\prime},t)$ scalar MSCR code from the construction in Section \ref{sec4}. Thus the desired $(n,k,d,t)$ scalar MSCR code can be  constructed as in Theorem \ref{thm8}.
\end{proof}

As a result, when $2k-1-t\leq k$, i.e., $k\leq t+1$, our construction presents scalar MSCR codes for all $d\geq k$ which covers all possible parameters for $(n,k,d,t)$ cooperative regenerating codes. When $2k-1-t> k$, i.e., $k>t+1$, our construction restricts to the case $d\geq 2k-1-t$. As a complementary result, in the next section we will prove the nonexistence of a family of linear scalar MSCR codes for $k<d<2k-2-t$. Recall that, the existence of scalar MSCR codes for $d=k$ has been ensured in \cite{Shum2011}.

\section{Nonexistence of linear scalar MSCR codes for $k<d<2k-2-t$}\label{sec5-c}
The nonexistence result relies on an assumption that the linear MSCR codes have {\it invariant repair spaces}.  In the following, we first describe the linear model for MSCR codes and explain the property of invariant repair space. Then we derive the interference alignment property for such MSCR codes and prove the nonexistence result under the condition $k<d<2k-2-t$.

\subsection{Linear MSCR codes with invariant repair space}\label{sec6a}
Suppose there exists an $(n,k,d,t,\alpha,\beta)$ linear MSCR code $\mathcal{C}$ over $\mathbb{F}_q$. By Theorem \ref{thmthm5}, we can always assume that $\mathcal{C}$ is systematic and has the following generator matrix
\begin{equation}\label{eq101}
G=\begin{pmatrix}
I_{\alpha}&&&\\
&I_{\alpha}&&\\
&&\ddots&\\
&&&I_{\alpha}\\
A_{1,1}&A_{1,2}&\cdots&A_{1,k}\\
\vdots&&&\\
A_{r,1}&A_{r,2}&\cdots&A_{r,k}
\end{pmatrix},
\end{equation}
where $I_{\alpha}$ denotes the $\alpha\times\alpha$ identity matrix, $A_{i,j}$ is a $\alpha\times\alpha$ matrix over $\mathbb{F}_q$ for $i\in[r]$, $j\in[k]$, and $r=n-k$.
For simplicity, every $\alpha$ consecutive rows of $G$ are regarded as a thick row, thus $G$ has $n$ thick rows which exactly correspond to the $n$ storage nodes.
For any data vector $\bm{m}\in \mathbb{F}_{q}^{k\alpha}$, node $i$ stores an $\alpha$-dimensional vector $\bm{c}_{i}^{\tau}=G_{i}\bm{m}^{\tau}$, where $G_i$ denotes the $i$-th thick row of $G$ and $i\in[n]$. Since $\bm{m}$ is independently and uniformly chosen from $\mathbb{F}_q^{k\alpha}$, we can also view node $i$ as storing the linear space spanned by the rows of $G_{i}$, denoted by $\langle G_i\rangle$, which is a subspace of $\mathbb{F}_{q}^{k\alpha}$. Then the {\it data reconstruction} requirement can be restated as follows.

{\bf Data reconstruction.~}The subspaces stored in any $k$ nodes can generate the entire space, namely, for any $i_{1},\ldots,i_{k}\in[n]$, $\sum_{j=1}^{k}\langle G_{i_{j}}\rangle=\mathbb{F}_{q}^{k\alpha}$.

Obviously, the data reconstruction requirement implies that each $A_{i,j}$ is invertible for all $i\in[r]$ and $j\in[k]$.

Then we describe the cooperative repair process. Suppose $\mathcal{F}\subset[n]$ is the set of failed nodes and $|\mathcal{F}|=t$. For each $i\in\mathcal{F}$, let $\mathcal{H}_i\subseteq [n]\setminus\mathcal{F}$ denote the set of helper nodes for repairing node $i$ and $|\mathcal{H}_i|=d$. In the first phase of the repair process, each node $j\in\mathcal{H}_i$ transmits $S_{j\rightarrow i, \mathcal{F},\mathcal{H}}\bm{c}_{j}^\tau=S_{j\rightarrow i, \mathcal{F},\mathcal{H}}G_{j}\bm{m}^{\tau}$ to repair node $i$, where
$S_{j\rightarrow i, \mathcal{F},\mathcal{H}}$ is a $\beta_1\times\alpha$ matrix corresponding to the linear transformation performed on node $j$ and $\mathcal{H}=(\mathcal{H}_i)_{i\in\mathcal{F}}$. From the view of linear spaces, we call $\langle S_{j\rightarrow i, \mathcal{F},\mathcal{H}}G_j\rangle$ as the {\it repair space} of node $j$ for repairing node $i$ with respect to the failed node set $\mathcal{F}$ and the helper node set $\mathcal{H}$. In the second phase of the repair process, the nodes in $\mathcal{F}$ exchange data with each other. Specifically, for any node $i\in\mathcal{F}$, each node $i'\in\mathcal{F}\setminus\{i\}$ transmits to node $i$ a $\beta_2$-dimensional repair space $\langle \gamma_{i'\rightarrow i, \mathcal{F},\mathcal{H}}\rangle$, where $\gamma_{i'\rightarrow i, \mathcal{F},\mathcal{H}}$ is a $\beta_2\times k\alpha$ matrix generated from $\sum_{j\in\mathcal{H}_{i'}}\langle S_{j\rightarrow i', \mathcal{F},\mathcal{H}}G_{j}\rangle$. Then the {\it node repair} requirement can be restated as follows.

{\bf Cooperative repair.~}For any $i\in\mathcal{F}$, the space stored by node $i$ can be recovered from the repair spaces collected by node $i$ in the two phases of the cooperative repair process, i.e., $\langle G_{i}\rangle\subseteq \sum_{j\in\mathcal{H}_{i}}\langle S_{j\rightarrow i, \mathcal{F},\mathcal{H}}G_{j}\rangle+\sum_{i'\in\mathcal{F}\setminus\{i\}}\langle\gamma_{i'\rightarrow i, \mathcal{F},\mathcal{H}}\rangle$.

\begin{definition}
A linear MSCR code with the generator matrix defined in (\ref{eq101}) is said to have invariant repair spaces if for any $i,j\in[n]$, the repair space $\langle S_{j\rightarrow i, \mathcal{F},\mathcal{H}}G_j\rangle$ is independent of $\mathcal{F}$ and $\mathcal{H}$, or equivalently, the repair matrix $S_{j\rightarrow i, \mathcal{F},\mathcal{H}}$ is independent of $\mathcal{F}$ and $\mathcal{H}$.
\end{definition}

As a result, for a linear MSCR code with invariant repair spaces, we denote the repair matrix by $S_{j\rightarrow i}$ instead of $S_{j\rightarrow i, \mathcal{F},\mathcal{H}}$, which means so long as node $j$ is connected to repair node $i$ in a cooperative repair of $t$ failed nodes containing $i$, node $j$ always performs the same linear transformation $S_{j\rightarrow i}$ on its stored data in spite of the identity of other failed nodes and other helper nodes. This property brings great convenience to the repair process in practice. Actually, most of the existing scalar MSCR codes have invariant repair spaces, such as the codes constructed in \cite{Chen2013,Shum2016,Li2014}. Moreover, the construction in this paper also satisfies this property.

\subsection{Interference alignment and nonexistence result}\label{sec6b}
Next we consider only linear scalar MSCR codes that have invariant repair spaces. Since for scalar MSCR codes, it holds $\beta_1=\beta_2=1$ and $\alpha=d-k+t$, the repair spaces $\langle S_{j\rightarrow i}G_j\rangle$ and $\langle\gamma_{i'\rightarrow i, \mathcal{F},\mathcal{H}}\rangle$ are both $1$-dimensional subspaces, thus are denoted as $\langle {\bm s}_{j\rightarrow i}G_j\rangle$ and $\langle{\bm\gamma}_{i'\rightarrow i, \mathcal{F},\mathcal{H}}\rangle$ respectively. Moreover, we always assume the scalar MSCR code is systematic and has a generator matrix as defined in (\ref{eq101}). Note that the property of having invariant repair spaces is maintained after transforming a linear MSCR code to a systematic one as illustrated in Theorem \ref{thmthm5}.

\begin{lemma}\label{lem8}(Interference alignment) Given a linear scalar MSCR code as described above, then we have,

$(a)$ For any $i,j\in[k]$, $i\neq j$,
  \begin{equation*}
  \bm{s}_{k+1\rightarrow i}A_{1,j}\sim\bm{s}_{k+2\rightarrow i}A_{2,j}\sim\cdots\sim \bm{s}_{k+\alpha\rightarrow i}A_{\alpha,j},
  \end{equation*}
  where for two vectors ${\bm x}$ and ${\bm x}'$, the notation $\bm{x}\sim \bm{x}'$ means $\bm{x}=c\bm{x}'$ for some nonzero $c\in \mathbb{F}_{q}^{*}$.

$(b)$ Suppose $k\geq t$, then for any $i\in[k]$ and any $j_{1},\ldots,j_{t-1}\in[k]\setminus \{i\}$,
  \begin{equation*}
  {\rm rank}\begin{pmatrix}
  \bm{s}_{k+1\rightarrow i}A_{1,i} \\ \vdots \\ \bm{s}_{k+\alpha\rightarrow i}A_{\alpha,i} \\ \bm{\gamma}^{(i)}_{j_{1}\rightarrow i,\mathcal{F},\mathcal{H}} \\ \vdots \\ \bm{\gamma}^{(i)}_{j_{t-1}\rightarrow i,\mathcal{F},\mathcal{H}}
  \end{pmatrix}=\alpha,
  \end{equation*}
  where $\bm{\gamma}^{(i)}_{j_{l}\rightarrow i,\mathcal{F},\mathcal{H}}\in\mathbb{F}_q^{\alpha}$ denotes the $i$-th component of $\bm{\gamma}_{j_{l}\rightarrow i,\mathcal{F},\mathcal{H}}$, that is, $\bm{\gamma}_{j_{l}\rightarrow i,\mathcal{F},\mathcal{H}}=(\bm{\gamma}^{(1)}_{j_{l}\rightarrow i,\mathcal{F},\mathcal{H}},...,\bm{\gamma}^{(k)}_{j_{l}\rightarrow i,\mathcal{F},\mathcal{H}})$, $1\leq l\leq t-1$, $\mathcal{F}=\{i,j_1,...,j_{t-1}\}\subseteq[k]$ and $\mathcal{H}_{i'}=[k+\alpha]\setminus\mathcal{F}$ for all $i'\in\mathcal{F}$.

\end{lemma}

\begin{proof}
Note that $k+\alpha=k+(d-k+t)=d+t\leq n$ and $\alpha=d-k+t\geq t$. The lemma is proved by considering the node repair requirement in different repair patterns.

(a) Let $\mathcal{F}=\{1,k+\alpha-t+2,\ldots,k+\alpha\}$, and $\mathcal{H}_{i}=\{2,\ldots,k\}\cup\{k+1,\ldots,k+\alpha-t+1\}$ for all $i\in\mathcal{F}$. That is, one systematic node and $t-1$ parity nodes fail, and the remaining $k-1$ systematic nodes and other $d-k+1$ parity nodes are helper nodes. Then after the repair process, node $1$ collects the space
\[
\Omega_{1}=\begin{pmatrix}
&\bm{s}_{2\rightarrow1}&& \\
&&\ddots& \\
&&&\bm{s}_{k\rightarrow1} \\
\bm{s}_{k+1\rightarrow1}A_{1,1}&\bm{s}_{k+1\rightarrow1}A_{1,2}&\cdots&\bm{s}_{k+1\rightarrow1}A_{1,k} \\
\vdots&\vdots&\vdots&\vdots \\
\bm{s}_{k+\alpha-t+1\rightarrow1}A_{\alpha-t+1,1}&\bm{s}_{k+\alpha-t+1\rightarrow1}A_{\alpha-t+1,2}&\cdots&\bm{s}_{k+\alpha-t+1\rightarrow1}A_{\alpha-t+1,k} \\
\bm{\gamma}_{k+\alpha-t+2\rightarrow1,\mathcal{F},\mathcal{H}}^{(1)}&\bm{\gamma}_{k+\alpha-t+2\rightarrow1,\mathcal{F},\mathcal{H}}^{(2)}&\cdots&\bm{\gamma}_{k+\alpha-t+2\rightarrow1,\mathcal{F},\mathcal{H}}^{(k)} \\
\vdots&\vdots&\vdots&\vdots \\
\bm{\gamma}_{k+\alpha\rightarrow1,\mathcal{F},\mathcal{H}}^{(1)}&\bm{\gamma}_{k+\alpha\rightarrow1,\mathcal{F},\mathcal{H}}^{(2)}&\cdots&\bm{\gamma}_{k+\alpha\rightarrow1,\mathcal{F},\mathcal{H}}^{(k)}
\end{pmatrix}.
\]
The node repair requirement implies that there exists an $\alpha\times(d+t-1)$ matrix $B=(\bm{b}_{2}^\tau\ \cdots \  \bm{b}_{k+\alpha}^\tau)$ such that
\[
B\Omega_{1}=G_{1}=\begin{pmatrix}
I_\alpha&0&\cdots&0
\end{pmatrix},
\]
where each $\bm{b}_{i}^\tau$ is an $\alpha\times1$ column vector for $2\leq i\leq k+\alpha$. More specifically,
\begin{equation}\label{*}
\begin{pmatrix}
\bm{b}_{k+1}^\tau&\cdots&\bm{b}_{k+\alpha-t+1}^\tau&\bm{b}_{k+\alpha-t+2}^\tau& \cdots& \bm{b}_{k+\alpha}^\tau
\end{pmatrix}\begin{pmatrix}
\bm{s}_{k+1\rightarrow1}A_{1,1}\\
\vdots\\
\bm{s}_{k+\alpha-t+1\rightarrow1}A_{\alpha-t+1,1}\\
\bm{\gamma}_{k+\alpha-t+2\rightarrow1,\mathcal{F},\mathcal{H}}^{(1)}\\
\vdots\\
\bm{\gamma}_{k+\alpha\rightarrow1,\mathcal{F},\mathcal{H}}^{(1)}
\end{pmatrix}=I_\alpha,
\end{equation}
and for $j\in\{2,\ldots,k\}$,
\begin{equation}\label{**}
\begin{pmatrix}
\bm{b}_{j}^\tau&\bm{b}_{k+1}^\tau&\cdots&\bm{b}_{k+\alpha-t+1}^\tau&\bm{b}_{k+\alpha-t+2}^\tau& \cdots& \bm{b}_{k+\alpha}^\tau
\end{pmatrix}\begin{pmatrix}
\bm{s}_{j\rightarrow1}\\
\bm{s}_{k+1\rightarrow1}A_{1,j}\\
\vdots\\
\bm{s}_{k+\alpha-t+1\rightarrow1}A_{\alpha-t+1,j}\\
\bm{\gamma}_{k+\alpha-t+2\rightarrow1,\mathcal{F},\mathcal{H}}^{(j)}\\
\vdots\\
\bm{\gamma}_{k+\alpha\rightarrow1,\mathcal{F},\mathcal{H}}^{(j)}
\end{pmatrix}=0.
\end{equation}
It follows from (\ref{*}) that
\[\rm rank
\begin{pmatrix}
\bm{b}_{k+1}^\tau&\cdots&\bm{b}_{k+\alpha-t+1}^\tau&\bm{b}_{k+\alpha-t+2}^\tau& \cdots& \bm{b}_{k+\alpha}^\tau
\end{pmatrix}=\alpha.
\]
As a result, the matrix
\[
\begin{pmatrix}
\bm{b}_{j}^\tau&\bm{b}_{k+1}^\tau&\cdots&\bm{b}_{k+\alpha-t+1}^\tau&\bm{b}_{k+\alpha-t+2}^\tau& \cdots& \bm{b}_{k+\alpha}^\tau
\end{pmatrix}
\]
is an $\alpha\times(\alpha+1)$ matrix of rank $\alpha$. Then from the equality (\ref{**}) we have for $2\leq j\leq k$,
\[
\rm rank
\begin{pmatrix}
\bm{s}_{j\rightarrow1}\\
\bm{s}_{k+1\rightarrow1}A_{1,j}\\
\vdots\\
\bm{s}_{k+\alpha-t+1\rightarrow1}A_{\alpha-t+1,j}\\
\bm{\gamma}_{k+\alpha-t+2\rightarrow1,\mathcal{F},\mathcal{H}}^{(j)}\\
\vdots\\
\bm{\gamma}_{k+\alpha\rightarrow1,\mathcal{F},\mathcal{H}}^{(j)}
\end{pmatrix}\leq1.
\]
Therefore,
$$
\bm{s}_{j\rightarrow1}\sim \bm{s}_{k+1\rightarrow1}A_{1,j}\sim \cdots\sim \bm{s}_{k+\alpha-t+1\rightarrow1}A_{\alpha-t+1,j}.
$$
Consider a new repair pattern  by exchanging the positions of node $k+1$ and $k+\alpha-t+2$, i.e. let $k+1$ be a failed node and $k+\alpha-t+2$ be a helper node, then we can derive $\bm{s}_{j\rightarrow1}\sim \bm{s}_{k+\alpha-t+2\rightarrow1}A_{\alpha-t+2,j}$, for $j\in\{2,\ldots,k\}$. Continue this way, then we  finally obtain
\[
\bm{s}_{j\rightarrow1}\sim \bm{s}_{k+1\rightarrow1}A_{1,j}\sim \cdots\sim \bm{s}_{k+\alpha\rightarrow1}A_{\alpha,j},\ \  2\leq j\leq k.
\]
Substituting node $1$ by an arbitrary $i\in[k]$, then (a) is obtained.

(b) Without loss of generality, suppose $i=1$ and $\{j_{1},\ldots,j_{t-1}\}=\{2,\ldots,t\}$. With respect to the $\mathcal{F}$ and $\mathcal{H}$ defined in (b), node $1$ receives the repair space
\[
\Omega_{2}=\begin{pmatrix}
&&&&\bm{s}_{t+1\rightarrow1}&& \\
&&&&&\ddots& \\
&&&&&&\bm{s}_{k\rightarrow1} \\
\bm{s}_{k+1\rightarrow1}A_{1,1}&\bm{s}_{k+1\rightarrow1}A_{1,2}&\cdots&\bm{s}_{k+1\rightarrow1}A_{1,t}&\bm{s}_{k+1\rightarrow1}A_{1,t+1}&\cdots&\bm{s}_{k+1\rightarrow1}A_{1,k} \\
\vdots&\vdots&\vdots&\vdots&\vdots&\vdots&\vdots \\
\bm{s}_{k+\alpha\rightarrow1}A_{\alpha,1}&\bm{s}_{k+\alpha\rightarrow1}A_{\alpha,2}&\cdots&\bm{s}_{k+\alpha\rightarrow1}A_{\alpha,t}&\bm{s}_{k+\alpha\rightarrow1}A_{\alpha,t+1}&\cdots&\bm{s}_{k+\alpha\rightarrow1}A_{\alpha,k} \\
\bm{\gamma}_{2\rightarrow1,\mathcal{F},\mathcal{H}}^{ (1)}&\bm{\gamma}_{2\rightarrow1,\mathcal{F},\mathcal{H}}^{ (2)}&\cdots&\bm{\gamma}_{2\rightarrow1,\mathcal{F},\mathcal{H}}^{ (t)}&\bm{\gamma}_{2\rightarrow1,\mathcal{F},\mathcal{H}}^{ (t+1)}&\cdots&\bm{\gamma}_{2\rightarrow1,\mathcal{F},\mathcal{H}}^{ (k)} \\
\vdots&\vdots&\vdots&\vdots&\vdots&\vdots&\vdots \\
\bm{\gamma}_{t\rightarrow1,\mathcal{F},\mathcal{H}}^{(1)}&\bm{\gamma}_{t\rightarrow1,\mathcal{F},\mathcal{H}}^{(2)}&\cdots&\bm{\gamma}_{t\rightarrow1,\mathcal{F},\mathcal{H}}^{(t)}&\bm{\gamma}_{t\rightarrow1,\mathcal{F},\mathcal{H}}^{(t+1)}&\cdots&\bm{\gamma}_{t\rightarrow1,\mathcal{F},\mathcal{H}}^{(k)}
\end{pmatrix}.
\]
Since $\langle G_{1}\rangle\subseteq\langle\Omega_{2}\rangle$, it is easy to see that (b) holds.
\end{proof}

\begin{lemma}\label{lem9}
Suppose $d\leq 2k-1-t$, then for any $p\in\{k+1,\ldots,k+\alpha\}$, any $\alpha$ out of the $k$ vectors $\{\bm{s}_{p\rightarrow1},\ldots,\bm{s}_{p\rightarrow k}\}$ are linearly independent.
\end{lemma}

\begin{proof}
Assume on the contrary that for some $p\in\{k+1,\ldots,k+\alpha\}$, there exist $\alpha$ linearly dependent vectors in $\{\bm{s}_{p\rightarrow1},\ldots,\bm{s}_{p\rightarrow k}\}$.
Without loss of generality, assume $p=k+1$ and \begin{equation}\label{eq102}\bm{s}_{k+1\rightarrow1}\in \langle \bm{s}_{k+1\rightarrow2},\ldots,\bm{s}_{k+1\rightarrow\alpha}\rangle,\end{equation} where the notation $\langle \bm{s}_{k+1\rightarrow2},\ldots,\bm{s}_{k+1\rightarrow\alpha}\rangle$ denotes the space spanned by $\{\bm{s}_{k+1\rightarrow2},\ldots,\bm{s}_{k+1\rightarrow\alpha}\}$. In the following, we will show the linear dependence in (\ref{eq102}) can be extended to all $k+j$ for $1\leq j\leq \alpha$, which will then lead to a contradiction to Lemma \ref{lem8} (b).

Since $d\leq 2k-1-t$, then $k\geq d-k+t+1=\alpha+1$, which means that there exists the $(\alpha+1)$-th component. Multiplying the invertible matrix $A_{1,\alpha+1}$ on both sides of (\ref{eq102}), we have
\[
\bm{s}_{k+1\rightarrow1}A_{1,\alpha+1}\in \langle \bm{s}_{k+1\rightarrow2}A_{1,\alpha+1},\ldots,\bm{s}_{k+1\rightarrow\alpha}A_{1,\alpha+1}\rangle.
\]
By Lemma \ref{lem8} (a), we can further obtain that  for $j\in\{1,\ldots,\alpha\}$,
\begin{equation*}
\bm{s}_{k+j\rightarrow1}A_{j,\alpha+1}\in \langle \bm{s}_{k+j\rightarrow2}A_{j,\alpha+1},\ldots,\bm{s}_{k+j\rightarrow\alpha}A_{j,\alpha+1}\rangle.
\end{equation*}
Note that all the $A_{j,\alpha+1}$'s are invertible from the data reconstruction requirement. Multiplying the inverse $A_{j,\alpha+1}^{-1}$, we have
\begin{equation}\label{***}
\bm{s}_{k+j\rightarrow1}\in \langle \bm{s}_{k+j\rightarrow2},\ldots,\bm{s}_{k+j\rightarrow\alpha}\rangle,~1\leq j\leq \alpha.
\end{equation}

From $d\leq 2k-1-t$ it also follows $k\geq d-k+t+1>t$, thus consider the repair pattern $\mathcal{F}=[t]\subset[k]$ and $\mathcal{H}_i=[k+\alpha]\setminus\mathcal{F}$ for all $i\in\mathcal{F}$. From Lemma \ref{lem8} (b), we know $\langle\{\bm{s}_{k+j\rightarrow1}A_{j,1}| 1\leq j \leq \alpha\} \cup\{\bm{\gamma}^{(1)}_{2\rightarrow1,\mathcal{F},\mathcal{H}},\ldots,\bm{\gamma}^{(1)}_{t\rightarrow1,\mathcal{F},\mathcal{H}}\}\rangle$ has dimension $\alpha$. However, by the definition of $\bm{\gamma}_{2\rightarrow1,\mathcal{F},\mathcal{H}}^{(1)}$, it has
$\bm{\gamma}^{(1)}_{2\rightarrow1,\mathcal{F},\mathcal{H}}\in\langle \bm{s}_{k+1\rightarrow2}A_{1,1},\ldots,\bm{s}_{k+\alpha\rightarrow2}A_{\alpha,1}\rangle=\langle \bm{s}_{k+1\rightarrow2}A_{1,1}\rangle$ where the equality follows from Lemma \ref{lem8} (a). In a similar way,  we get $\bm{\gamma}^{(1)}_{i\rightarrow1,\mathcal{F},\mathcal{H}}\in\langle \bm{s}_{k+1\rightarrow i}A_{1,1},\ldots,\bm{s}_{k+\alpha\rightarrow i}A_{\alpha,1}\rangle=\langle \bm{s}_{k+1\rightarrow i}A_{1,1}\rangle$ for $2\leq i\leq t$.
Therefore,
\begin{eqnarray}
&&\langle\{\bm{s}_{k+j\rightarrow1}A_{j,1}| 1\leq j \leq \alpha\} \cup\{\bm{\gamma}^{(1)}_{2\rightarrow1,\mathcal{F},\mathcal{H}},\ldots,\bm{\gamma}^{(1)}_{t\rightarrow1,\mathcal{F},\mathcal{H}}\}\rangle \nonumber\\
&{\subseteq}&\langle\{\bm{s}_{k+j\rightarrow2}A_{j,1},\ldots,\bm{s}_{k+j\rightarrow\alpha}A_{j,1}| 1\leq j \leq \alpha\} \cup\{\bm{s}_{k+1\rightarrow2}A_{1,1},\ldots,\bm{s}_{k+1\rightarrow t}A_{1,1}\}\rangle\label{lem8-1}\\
&{\subseteq}&\langle\{\bm{s}_{k+1\rightarrow2}A_{1,1},\ldots,\bm{s}_{k+1\rightarrow\alpha}A_{1,1}\} \rangle\label{lem8-2},
\end{eqnarray}
where (\ref{lem8-1}) comes from (\ref{***}), and (\ref{lem8-2}) comes from Lemma \ref{lem8} (a) and the fact that $\alpha=d-k+t\geq t$. So, (\ref{lem8-2}) implies that $\langle\{\bm{s}_{k+j\rightarrow1}A_{j,1}| 1\leq j \leq \alpha\} \cup\{\bm{\gamma}^{(1)}_{2\rightarrow1,\mathcal{F},\mathcal{H}},\ldots,\bm{\gamma}^{(1)}_{t\rightarrow1,\mathcal{F},\mathcal{H}}\}\rangle$ has dimension at most $\alpha-1$ which contradicts to Lemma \ref{lem8} (b).
\end{proof}

\begin{theorem}
For $k<d<2k-2-t$, there exist no linear scalar $(n,k,d,t)$ MSCR codes that have invariant repair spaces.
\end{theorem}

\begin{proof}
Assume on the contrary there exists such an MSCR code for some $n,k,d,t$ with $k<d<2k-2-t$. As stated before, by Theorem \ref{thmthm5} we can always assume this MSCR code is systematic and has a generator matrix as defined in (\ref{eq101}). Since $d<2k-2-t$, it follows $k>d-k+t+2=\alpha+2$, i.e. $k\geq\alpha+3>t$.

We first consider the repair of node $\alpha+1$ and node $\alpha+2$ by the helper nodes $k+j$, $1\leq j\leq\alpha$.
In particular, we restrict to the $(\alpha+2)$-th and the $(\alpha+3)$-th components. That is, by Lemma \ref{lem8} (a), we have
 \begin{equation}\label{inter3}
 \bm{s}_{k+1\rightarrow\alpha+1}A_{1,\alpha+2}\sim\bm{s}_{k+2\rightarrow\alpha+1}A_{2,\alpha+2}\sim\cdots\sim \bm{s}_{k+\alpha\rightarrow\alpha+1}A_{\alpha,\alpha+2},
 \end{equation}
 \begin{equation}\label{inter4}
 \bm{s}_{k+1\rightarrow\alpha+1}A_{1,\alpha+3}\sim \bm{s}_{k+2\rightarrow\alpha+1}A_{2,\alpha+3}\sim\cdots\sim \bm{s}_{k+\alpha\rightarrow\alpha+1}A_{\alpha,\alpha+3},
 \end{equation}
\begin{equation}\label{inter5}
 \bm{s}_{k+1\rightarrow\alpha+2}A_{1,\alpha+3}\sim\bm{s}_{k+2\rightarrow\alpha+2}A_{2,\alpha+3}\sim\cdots\sim \bm{s}_{k+\alpha\rightarrow\alpha+2}A_{\alpha,\alpha+3}.
 \end{equation}
Then our proof goes along the following line. First, represent both $\bm{s}_{k+j\rightarrow\alpha+1}$ and $\bm{s}_{k+j\rightarrow\alpha+2}$ as linear combinations of $\{\bm{s}_{k+j\rightarrow 1},...,\bm{s}_{k+j\rightarrow \alpha}\}$ (by Lemma \ref{lem9}). Then by Lemma \ref{lem8} (a) and (\ref{inter3}), (\ref{inter4}) we can derive a relation between the $(\alpha+2)$-th and the $(\alpha+3)$-th components. With this relation we can finally substitute the $A_{j,\alpha+3}$'s in (\ref{inter5}) with $A_{j,\alpha+2}$'s and then obtain a contradiction to Lemma \ref{lem8} (b). The details are as follows.

From Lemma \ref{lem9}, we know that for $1\leq j\leq \alpha$, the $\alpha$ vectors $\bm{s}_{k+j\rightarrow1},\ldots,\bm{s}_{k+j\rightarrow\alpha}$ each of length $\alpha$  are linearly independent, thus $\bm{s}_{k+j\rightarrow\alpha+1}$ can be represented as a linear combination of $\{\bm{s}_{k+j\rightarrow1},\ldots,\bm{s}_{k+j\rightarrow\alpha}\}$. Specifically, suppose
 $$\bm{s}_{k+j\rightarrow\alpha+1}=\bm{\lambda}_{k+j,\alpha+1}\begin{pmatrix}
  \bm{s}_{k+j\rightarrow1}\\ \vdots \\ \bm{s}_{k+j\rightarrow\alpha}
  \end{pmatrix}=\begin{pmatrix}\lambda_{k+j,\alpha+1}^{(1)} & \cdots &\lambda_{k+j,\alpha+1}^{(\alpha)}\end{pmatrix}\begin{pmatrix}
  \bm{s}_{k+j\rightarrow1}\\ \vdots \\ \bm{s}_{k+j\rightarrow\alpha}
  \end{pmatrix},$$
  where $\bm{\lambda}_{k+j,\alpha+1}=(\lambda_{k+j,\alpha+1}^{(1)} \ \cdots \ \lambda_{k+j,\alpha+1}^{(\alpha)})\in\mathbb{F}_q^\alpha$. Moreover, we claim that $\lambda_{k+j,\alpha+1}^{(i)}\neq0$ for $1\leq i\leq \alpha$. Otherwise, it leads to a contradiction to Lemma \ref{lem9}.
Similarly, for $1\leq j\leq\alpha$, we can  write $\bm{s}_{k+j\rightarrow\alpha+2}$ as
  $$\bm{s}_{k+j\rightarrow\alpha+2}=\bm{\lambda}_{k+j,\alpha+2}\begin{pmatrix}
  \bm{s}_{k+j\rightarrow1}\\ \vdots \\ \bm{s}_{k+j\rightarrow\alpha}
  \end{pmatrix}=\begin{pmatrix}\lambda_{k+j,\alpha+2}^{(1)} & \cdots &\lambda_{k+j,\alpha+2}^{(\alpha)}\end{pmatrix}\begin{pmatrix}
  \bm{s}_{k+j\rightarrow1}\\ \vdots \\ \bm{s}_{k+j\rightarrow\alpha}
  \end{pmatrix},$$
  where $\bm{\lambda}_{k+j,\alpha+2}\in(\mathbb{F}_q^*)^\alpha$.

  For simplicity, we denote for $1\leq j\leq\alpha$,
  \[
  B_{j}=\begin{pmatrix}
  \bm{s}_{k+j\rightarrow1}\\ \vdots \\ \bm{s}_{k+j\rightarrow\alpha}
  \end{pmatrix}A_{j,\alpha+2}=\begin{pmatrix}
  \bm{s}_{k+j\rightarrow1}A_{j,\alpha+2}\\ \vdots \\ \bm{s}_{k+j\rightarrow\alpha}A_{j,\alpha+2}
  \end{pmatrix},
  \]
  \[
  C_{j}=\begin{pmatrix}
  \bm{s}_{k+j\rightarrow1}\\ \vdots \\ \bm{s}_{k+j\rightarrow\alpha}
  \end{pmatrix}A_{j,\alpha+3}=\begin{pmatrix}
  \bm{s}_{k+j\rightarrow1}A_{j,\alpha+3}\\ \vdots \\ \bm{s}_{k+j\rightarrow\alpha}A_{j,\alpha+3}
  \end{pmatrix}.
  \]
  Obviously, all the $B_{j}, C_{j}$'s are invertible, and (\ref{inter3})-(\ref{inter5}) can be respectively rewritten as below.
  \begin{equation}\label{16}
  \bm{\lambda}_{k+1,\alpha+1}B_{1}\sim\bm{\lambda}_{k+2,\alpha+1}B_{2}\sim\cdots\sim \bm{\lambda}_{k+\alpha,\alpha+1}B_{\alpha},
  \end{equation}
  \begin{equation}\label{17}
  \bm{\lambda}_{k+1,\alpha+1}C_{1}\sim\bm{\lambda}_{k+2,\alpha+1}C_{2}\sim\cdots\sim \bm{\lambda}_{k+\alpha,\alpha+1}C_{\alpha},
  \end{equation}
  \begin{equation}\label{18}
  \bm{\lambda}_{k+1,\alpha+2}C_{1}\sim\bm{\lambda}_{k+2,\alpha+2}C_{2}\sim\cdots\sim \bm{\lambda}_{k+\alpha,\alpha+2}C_{\alpha}.
  \end{equation}

  From Lemma \ref{lem8} (a), we know that  there exist diagonal matrices $\Lambda_{j}$ and $\Gamma_{j}$ such that $B_{j}=\Lambda_{j}B_{1}$ and $C_{j}=\Gamma_{j}C_{1}$ for $ 2\leq j\leq\alpha$. Since $B_{1}$ and $C_{1}$ are invertible, it follows from (\ref{16})-(\ref{18}) that
  \begin{equation}\label{out1}
  \bm{\lambda}_{k+1,\alpha+1}\sim\bm{\lambda}_{k+2,\alpha+1}\Lambda_{2}\sim\cdots\sim \bm{\lambda}_{k+\alpha,\alpha+1}\Lambda_{\alpha},
  \end{equation}
  \begin{equation}\label{out2}
  \bm{\lambda}_{k+1,\alpha+1}\sim\bm{\lambda}_{k+2,\alpha+1}\Gamma_{2}\sim\cdots\sim \bm{\lambda}_{k+\alpha,\alpha+1}\Gamma_{\alpha},
  \end{equation}
  \begin{equation}\label{out3}
  \bm{\lambda}_{k+1,\alpha+2}\sim\bm{\lambda}_{k+2,\alpha+2}\Gamma_{2}\sim\cdots\sim \bm{\lambda}_{k+\alpha,\alpha+2}\Gamma_{\alpha}.
  \end{equation}

  Then (\ref{out1}) and (\ref{out2}) imply that for $2\leq j\leq\alpha$,
  $$
  \bm{\lambda}_{k+j,\alpha+1}\Lambda_{j}\sim\bm{\lambda}_{k+j,\alpha+1}\Gamma_{j},
  $$
  i.e.
  $$
  \begin{pmatrix}\lambda_{k+j,\alpha+1}^{(1)} & \cdots &\lambda_{k+j,\alpha+1}^{(\alpha)}\end{pmatrix}\Lambda_{j}\sim \begin{pmatrix}\lambda_{k+j,\alpha+1}^{(1)} & \cdots &\lambda_{k+j,\alpha+1}^{(\alpha)}\end{pmatrix}\Gamma_{j}.
  $$
  Since all the components of $\bm{\lambda}_{k+j,\alpha+1}$ are nonzero,
  it follows
  \begin{equation}\label{inter0}
 \Lambda_{j}\sim\Gamma_{j}, {~\rm i.e.,~} \Lambda_{j}=c\Gamma_{j} {~\rm for~some~} c\in\mathbb{F}_q^*.
 \end{equation}

As a result, it follows from (\ref{out3}) and (\ref{inter0}) that
$\bm{\lambda}_{k+1,\alpha+2}\sim\bm{\lambda}_{k+2,\alpha+2}\Lambda_{2}\sim\cdots\sim\bm{\lambda}_{k+\alpha,\alpha+2}\Lambda_{\alpha}$.
Multiply each term by $B_{1}$ on the right, then we get $\bm{\lambda}_{k+1,\alpha+2}B_{1}\sim\bm{\lambda}_{k+2,\alpha+2}B_{2}\sim\cdots\sim\bm{\lambda}_{k+\alpha,\alpha+2}B_{\alpha}$,
i.e.,
  \begin{equation}\label{inter6}
  \bm{s}_{k+1\rightarrow\alpha+2}A_{1,\alpha+2}\sim\cdots\sim \bm{s}_{k+\alpha\rightarrow\alpha+2}A_{\alpha,\alpha+2}.
  \end{equation}
  However, by Lemma \ref{lem8} (b), we know that
  $$
  {\rm rank}\begin{pmatrix}
  \bm{s}_{k+1\rightarrow\alpha+2}A_{1,\alpha+2} \\ \vdots \\ \bm{s}_{k+\alpha\rightarrow\alpha+2}A_{\alpha,\alpha+2}
  \end{pmatrix}\geq\alpha-t+1=d-k+1\geq2 ,
  $$
  which contradicts to (\ref{inter6}). Thus the theorem is proved.
\end{proof}

\section{Conclusions}\label{sec6}
We explicitly construct scalar MSCR codes for all $d\geq\max\{2k-1-t,k\}$. The construction can be viewed as an extension of the product matrix code construction proposed in \cite{Kumar2011} for MSR and MBR codes. Just as in \cite{Kumar2011} where the product matrix-based MSR codes only applies when $d\geq 2k-2$, our construction of MSCR codes also restricts to $d\geq 2k-1-t$. Both restrictions lead to the same limit on the information rate, i.e., $\frac{k}{n}\leq\frac12+\frac{1}{2n}$. As complementary results, the nonexistence of certain scalar MSR codes for $k<d< 2k-3$ was presented in \cite{Shah2009} and the nonexistence of certain scalar MSCR codes for $k<d< 2k-t-2$ are given in this paper. Along with this work, several results achieved so far for cooperative regenerating codes can be seen as the counterparts of the corresponding results in regenerating codes, such as the cut-set bound (\cite{Dimakis2011} and \cite{Hu2010,Shum2011}) and the general construction of high-rate MSR codes and MSCR codes (\cite{Ye2016} and \cite{Ye2018}). On the one hand, both the parameter bound and the constructions for cooperative regenerating codes degenerate into those for regenerating codes. However, on the other hand, it is nontrivial to extend the results of regenerating codes to derive their counterparts in cooperative regenerating codes. An interesting question is how to generally build a cooperative regenerating code for repairing $t>1$ erasures from regenerating codes that are designed for repairing individual node failures.
Although we cannot solve this problem right now, we can predict that there should be more extensions in cooperative regenerating codes based on the fruitful research in regenerating codes.


\appendices
\section{Proof of Lemma~\ref{lem3}}\label{appenA}
\noindent {\bf Lemma 1.}
Let $\Phi$ be a $k\times(k-1)$ Vandermonde matrix defined as in (\ref{phi}), and $\Delta$ be a $k\times k$ diagonal matrix with distinct and nonzero diagonal elements. Suppose $$
X=\Phi S+\Delta\Phi T,
$$
where $S$ and $T$ are two $(k-1)\times(k-1)$ symmetric matrices. Then $S$ and $T$ can be uniquely computed from $X,\Phi$ and $\Delta$.

\begin{proof}
Multiply $\Phi S+\Delta\Phi T$ on the right side by $\Phi^\tau$, then we get
$$
X\Phi^\tau=\Phi S\Phi^\tau+\Delta \Phi T\Phi^\tau.
$$
Define
$$\begin{aligned}
&A=\Phi S\Phi^\tau \\
&B=\Phi T\Phi^\tau.
\end{aligned}$$
Since $S$ and $T$ are symmetric, then $A$ and $B$ are also symmetric. For $i,j\in[k]$, let $a_{ij}$ and $b_{ij}$ denote the $(i,j)$-entries of $A$ and $B$ respectively. And denote by $\lambda_{i}$ the $i$-th diagonal element of $\Delta,~i\in[k]$.

Then for $i,j\in[k], i\neq j$, the $(i,j)$-th entry of $A+\Delta B$ ($=X\Phi^\tau$) is
\begin{equation}\label{eq15}
a_{ij}+\lambda_{i}b_{ij},
\end{equation}
while the $(j,i)$-th entry of $A+\Delta B$ is
\begin{equation}\label{eq16}
a_{ji}+\lambda_{j}b_{ji}.
\end{equation}
Due to the symmetry of $A$ and $B$, it has $a_{ij}=a_{ji}, b_{ij}=b_{ji}$. Thus $a_{ij}$ and $b_{ij}$ can be solved from (\ref{eq15}) and (\ref{eq16}) since $\lambda_{i}\neq \lambda_{j}$ for $i\neq j$. Exhausting all $i,j\in[k], i\neq j$, then all the non-diagonal elements of $A$ and $B$ can be obtained.

Next we consider the matrix $A$. Fix some $i\in[k]$, consider the $k-1$ symbols $\{a_{ij}|j\in[k], j\neq i\}$. Each $a_{ij}$ is the product of the $i$-th row of $\Phi S$ and the $j$-th column of $\Phi^\tau$. Since the matrix containing any $k-1$ columns of $\Phi^\tau$ is a $(k-1)\times(k-1)$ Vandermonde matrix which is invertible, the $i$-th row of $\Phi S$ is solvable from $\{a_{ij}|j\in[k], j\neq i\}$. After collecting $k-1$ rows of $\Phi S$, the matrix $S$ can be recovered since any $k-1$ rows of $\Phi$ also form an invertible matrix.

In the same way,  the matrix $T$ can be obtained from non-diagonal elements of $B$.
\end{proof}

\end{document}